\documentclass[11pt]{article}

\usepackage[margin=0.8in]{geometry}
\usepackage{amsmath,amssymb,amsthm,mathtools}
\usepackage{xcolor}
\usepackage{enumitem}
\usepackage{microtype}
\usepackage{cite}

\newtheorem{theorem}{Theorem}
\newtheorem{proposition}[theorem]{Proposition}
\newtheorem{lemma}[theorem]{Lemma}
\newtheorem{corollary}[theorem]{Corollary}
\newtheorem{remark}[theorem]{Remark}

\title{
\textbf{Quantum Krylov Linear Solving Beyond Global Conditioning
via Spectral Compression and Directional Stability}
}

\author{
Sijia Yu \qquad Yifan Zhou\\[0.4em]
\small Department of Electrical and Computer Engineering,
Stony Brook University,\\
\small Stony Brook, NY 11794, USA
}

\date{}

\begin{document}
\maketitle

\begin{abstract}
Global condition-number dependence can substantially overestimate the
difficulty of preparing normalized quantum linear-system solutions.
Many existing beyond-conditioning approaches are favorable when globally
ill-conditioned directions have limited relevance to the target solution.
Here we study a harder regime for time-evolution quantum Krylov linear
solving (QKS), where a vanishing eigenvalue remains solution-relevant and
must be retained.
To address this difficulty, QKS compresses the solution-relevant spectrum
into a compact reduced space where the ill-conditioning is concentrated
along a soft direction aligned with the solution, so that normalization
removes the divergent amplification along that direction and leaves only
a bounded transverse response. Consequently, the normalized-state complexity need not inherit the
divergence of the global condition number.
Our contributions are threefold:
(1) We show that the solution-relevant information admits a compact and
stable QKS representation, with bounded subspace dimension, evolution time,
and reconstruction overhead as the relevant eigenvalue vanishes.
(2) We show that normalized-state difficulty is governed by whether inverse
amplification changes the solution direction, rather than by the small
eigenvalue alone; even an arbitrarily ill-conditioned reduced system can
remain directionally stable when the amplification is solution-aligned.
(3) We prove that this stability persists under finite projected-system
errors and transfers to the actual QKS output, yielding an end-to-end
complexity bound without inverse dependence on the retained small
eigenvalue.  An explicit separation family further shows that a full
inverse-amplification criterion can diverge while the corresponding QKS
quantities remain bounded.
\end{abstract}

\section{Introduction}
\label{sec:introduction}

Quantum linear-system (QLS) algorithms are a fundamental primitive for
quantum scientific computing, with applications to differential equations,
Green's-function evaluation, regression, optimization, and other numerical
problems involving large linear systems
~\cite{HHL,CKS,Chakraborty2019,Tong2021,Chakraborty2023,
Krovi2023,BerryCosta2024}.
Since the original HHL algorithm~\cite{HHL}, substantial progress has been
made in improving precision dependence, state-preparation overhead,
Hamiltonian simulation, and the treatment of positive-definite systems
~\cite{Ambainis2012,CKS,Subasi2019,AnLin2022,Costa2022,
Cunningham2024,Dalzell2024,LowSu2026}.
These developments are closely connected to advances in block encoding,
quantum signal processing, singular-value transformation, and polynomial
matrix functions
~\cite{LowChuang2017QSP,LowChuang2019Qubitization,Gilyen2019,
LinTong2020,Martyn2021,Motlagh2024}.
A central difficulty nevertheless remains: the complexity of QLS can grow
strongly with the condition number of the input matrix, and lower-bound
results show that such dependence cannot be removed in full generality
~\cite{HHL,CKS,Mori2026}.
For structured instances, however, the global condition number may be an
overly pessimistic description because the entire spectrum need not be
equally relevant to the particular normalized solution being prepared.

This observation has motivated several approaches that replace or weaken
global condition-number dependence by exploiting more favorable
instance-dependent structure
~\cite{Orsucci2021,Chakraborty2023,Li2025,Dalzell2026}.
In particular, recent beyond-condition-number results show that spectral
information relevant to a specific input and target task can lead to
complexity substantially smaller than that suggested by the worst-case
global condition number~\cite{Li2025,Dalzell2026}.
Related filtering, regularization, and polynomial-transformation techniques
likewise demonstrate that the difficulty of a linear-system instance can
depend strongly on how the input state is distributed across the spectrum
~\cite{LinTong2020,Gilyen2019,Chakraborty2023}.
A more challenging regime arises when the eigenvalue responsible for the
ill-conditioning is itself solution-relevant.  In this case, the small
eigenvalue cannot simply disappear from the relevant description: its
inverse contribution may dominate the solution and must be retained.
This raises the central question of this work:
\emph{can quantum linear solving remain efficient even when the eigenvalue
responsible for the diverging condition number remains essential to the
target solution?}

We study this question using time-evolution quantum Krylov-subspace (QKS)
linear solving.  Krylov-subspace methods are among the central tools of
classical numerical linear algebra, where low-dimensional subspaces generated
from repeated actions of an operator are used to approximate large
eigenvalue and linear-system problems
~\cite{Lanczos1950,HestenesStiefel1952,Paige1971,Greenbaum1997,Saad2003}.
Their quantum counterparts have recently been explored for eigensolvers,
subspace diagonalization, and linear-system solving
~\cite{Parrish2019,Epperly2022,Shen2023,Xu2024,Yoshioka2025,Yu2025}.
For QKS linear solving, starting from the right-hand side $|b\rangle$, one
generates time-evolved states
$e^{-ij\tau A}|b\rangle$ and uses their span as a reduced solution space
~\cite{Xu2024}.
The original linear system is projected onto this subspace, producing a
small classical system whose coefficients are subsequently used to
reconstruct the quantum solution.
This construction is particularly suitable for the present question because
the reduced space is generated from the input state itself and therefore
probes the spectrum through the amplitudes relevant to the target response,
rather than treating all spectral directions uniformly.
It also exposes the main channels through which ill-conditioning may enter
the end-to-end QKS complexity: the Krylov dimension and evolution time
required to represent the solution, the reconstruction probability, and the
sensitivity of the projected solve to finite estimation errors.

The key observation of this work is that these channels need not inherit
the divergence of a retained small eigenvalue.  First, when the normalized
solution is concentrated on a compact set of spectrally resolvable modes,
the solution-relevant inverse response can be compressed into a stable
time-evolution Krylov subspace with bounded dimension and evolution time.
More importantly, even when the reduced system explicitly retains a
vanishing eigenvalue and hence becomes arbitrarily ill conditioned, its
\emph{normalized} solution need not become unstable.  If the divergent
inverse amplification becomes increasingly aligned with the solution, it
primarily changes the magnitude of the response rather than its normalized
direction.  Normalization removes this solution-parallel amplification,
leaving only the transverse response that can change the quantum state.
Thus, in this regime, normalized-state difficulty is governed not by the
small eigenvalue alone, but by whether the resulting inverse amplification
changes the solution direction.

We make this mechanism precise by separating the structural and directional
parts of QKS.  First, we establish conditions under which solution-weighted
spectral concentration and bounded-time resolvability guarantee an accurate
QKS representation with bounded Krylov dimension, evolution horizon, and
state-reconstruction overhead.  Second, we analyze the projected problem in
equivalent orthonormal coordinates and consider a retained-soft-mode regime
in which one reduced eigenvalue approaches zero while the remaining reduced
spectrum stays separated from zero.  When the reduced right-hand side has
persistent overlap with this soft direction, the reduced inverse norm
diverges, yet the normalized directional sensitivity remains bounded.
We further show that this protection persists under finite projected-system
errors and transfers from the compact subspace representation to the actual
Galerkin output returned by QKS.  The finite-error analysis uses standard
spectral and invariant-subspace perturbation tools
~\cite{DavisKahan1970,StewartSun1990,Bhatia1997,Higham2008}.
Together, these results yield an end-to-end QKS complexity whose
condition-sensitive factors do not inherit the divergence of the retained
small eigenvalue, provided that the underlying matrix-access primitives do
not independently introduce such a dependence.

The resulting mechanism has both an overlap with and a distinction from
existing beyond-condition-number approaches.  When globally small
eigenvalues carry negligible weight in the target solution, QKS recovers a
familiar effective-conditioning picture: the relevant reduced problem can
itself remain well conditioned, consistent with the broader principle that
task-dependent spectral structure can weaken global condition-number
dependence~\cite{Orsucci2021,Li2025,Dalzell2026}.
The regime central to this work is different.  Here the small eigenvalue is
retained and remains solution-relevant, so the reduced inverse norm itself
still diverges.  Nevertheless, its divergent amplification can become
asymptotically parallel to the normalized solution and therefore have only a
bounded directional effect.  To make this distinction explicit, we
construct a family for which a full inverse-amplification criterion appearing
in recent beyond-conditioning analysis~\cite{Dalzell2026} diverges with the
small eigenvalue, whereas the corresponding QKS structural and directional
quantities remain bounded.  This separation shows that full inverse
amplification and normalized-state sensitivity are not equivalent measures
of linear-solving difficulty.

Our contributions are threefold.
(1) We show that the solution-relevant information admits a compact and
stable QKS representation, with bounded subspace dimension, evolution time,
and reconstruction overhead as the relevant eigenvalue vanishes.
(2) We show that normalized-state difficulty is governed by whether inverse
amplification changes the solution direction, rather than by the small
eigenvalue alone; even an arbitrarily ill-conditioned reduced system can
remain directionally stable when the amplification is solution-aligned.
(3) We prove that this stability persists under finite projected-system
errors and transfers to the actual QKS output, yielding an end-to-end
complexity bound without inverse dependence on the retained small
eigenvalue.  We further provide an explicit separation family for which a
full inverse-amplification criterion diverges while the corresponding QKS
quantities remain bounded.

The remainder of the paper is organized as follows.
Section~\ref{sec:qks_setup} introduces the QKS formulation and the
end-to-end complexity decomposition.
Section~\ref{sec:structural} establishes spectral compression and bounds the
structural resources.
Section~\ref{sec:directional} develops the directional-stability theory and
transfers the compact subspace approximation to the actual QKS output.
Section~\ref{sec:e2e_result} combines these ingredients into the end-to-end
beyond-conditioning result and instantiates it under several matrix-access
models.
Finally, Section~\ref{sec:separation} relates the mechanism to existing
beyond-condition-number approaches and establishes the explicit separation.

\section{Quantum Krylov Linear Solving and End-to-End Complexity}
\label{sec:qks_setup}

This section introduces the QKS problem, fixes the notation used throughout
the paper, and identifies the quantities that enter the end-to-end
complexity analysis.  We first define the linear-system problem and the
time-evolution Krylov subspace used by QKS.  We then describe the projected
system actually solved by the algorithm.  Finally, we introduce an
equivalent orthonormal coordinate system used only for the theoretical
analysis of normalized-state sensitivity.

\subsection{Linear-system problem and QKS representation}

We consider the Hermitian positive-definite linear system
\begin{equation}
    A x = |b\rangle,
    \label{eq:linear_system}
\end{equation}
where $A\in\mathbb C^{N\times N}$ is the system matrix and
$|b\rangle\in\mathbb C^N$ is the normalized right-hand side.  We assume
\begin{equation}
    A=A^\dagger\succ0,
    \qquad
    \|A\|_2\le1,
    \qquad
    \|b\|_2=1.
\end{equation}
The normalization $\|A\|_2\le1$ fixes the spectral scale and does not change
the normalized solution state.

Let
\begin{equation}
    A
    =
    \sum_{\ell=1}^{N}
    \lambda_\ell
    |u_\ell\rangle\langle u_\ell|,
    \label{eq:A_spectral_decomposition}
\end{equation}
where $\lambda_\ell>0$ is the $\ell$th eigenvalue of $A$ and
$|u_\ell\rangle$ is the corresponding orthonormal eigenvector.  Expanding
the right-hand side in this eigenbasis gives
\begin{equation}
    |b\rangle
    =
    \sum_{\ell=1}^{N}
    b_\ell |u_\ell\rangle,
    \label{eq:b_spectral_decomposition}
\end{equation}
where
$b_\ell=\langle u_\ell|b\rangle$
is the overlap of the input state with eigenmode $\ell$.

The exact unnormalized solution is therefore
\begin{equation}
    x
    :=
    A^{-1}|b\rangle
    =
    \sum_{\ell=1}^{N}
    \frac{b_\ell}{\lambda_\ell}
    |u_\ell\rangle.
    \label{eq:exact_solution}
\end{equation}
The quantum linear-system task is to prepare the corresponding normalized
state
\begin{equation}
    |x\rangle
    :=
    \frac{x}{\|x\|_2}.
    \label{eq:normalized_solution}
\end{equation}

QKS approximates this solution using time-evolved copies of the input
state.  Let $\tau>0$ denote the time step and let $m$ denote the number of
Krylov basis states.  The $j$th time-evolution state is
\begin{equation}
    |\phi_j\rangle
    :=
    e^{-ij\tau A}|b\rangle,
    \qquad
    j=0,\ldots,m-1.
    \label{eq:qks_basis_state}
\end{equation}
Collecting these states as columns gives the generally nonorthogonal basis
matrix
\begin{equation}
    \Phi_m
    :=
    [|\phi_0\rangle,\ldots,|\phi_{m-1}\rangle],
\end{equation}
and the associated QKS subspace is
\begin{equation}
    \mathcal K_m
    :=
    \operatorname{range}(\Phi_m).
    \label{eq:qks_subspace}
\end{equation}
Thus $m$ is the dimension of the candidate QKS representation, while
\begin{equation}
    T:=(m-1)\tau
    \label{eq:max_evolution_time}
\end{equation}
is the maximum real-time evolution required to generate the basis.

\subsection{Projected system solved by QKS}

The QKS approximation is written as
\begin{equation}
    x_m
    :=
    \Phi_m\boldsymbol\gamma,
    \label{eq:qks_output}
\end{equation}
where
$\boldsymbol\gamma\in\mathbb C^m$
contains the coefficients of the Krylov expansion.

Galerkin projection determines these coefficients from
\begin{equation}
    H_m\boldsymbol\gamma
    =
    \mathbf g_m,
    \label{eq:qks_projected}
\end{equation}
where
\begin{equation}
    H_m
    :=
    \Phi_m^\dagger A\Phi_m
\end{equation}
is the $m\times m$ projected system matrix, and
\begin{equation}
    \mathbf g_m
    :=
    \Phi_m^\dagger|b\rangle
\end{equation}
is the projected right-hand side.

Thus the QKS algorithm estimates the projected quantities
$H_m$ and $\mathbf g_m$, solves the reduced classical system
Eq.~\eqref{eq:qks_projected}, and reconstructs the state
$x_m=\Phi_m\boldsymbol\gamma$.

Because the time-evolution basis is generally nonorthogonal, we also define
its Gram matrix
\begin{equation}
    S_m
    :=
    \Phi_m^\dagger\Phi_m.
    \label{eq:gram_matrix}
\end{equation}
The Gram matrix measures the linear independence and geometry of the QKS
basis.  In the present analysis, $S_m$ is used to characterize basis
stability, reconstruction probability, and an equivalent orthonormal
coordinate system; it is not an additional reduced system that the QKS
algorithm must solve.

\subsection{Equivalent orthonormal coordinates for analysis}

When $S_m\succ0$, the columns of $\Phi_m$ are linearly independent.
For theoretical analysis, we define the orthonormal basis
\begin{equation}
    Q_m
    :=
    \Phi_mS_m^{-1/2}.
    \label{eq:orthonormal_basis}
\end{equation}
The matrix $Q_m$ spans exactly the same subspace as $\Phi_m$ but satisfies
\begin{equation}
    Q_m^\dagger Q_m=I.
\end{equation}

In this orthonormal coordinate system, define
\begin{equation}
    B_m
    :=
    S_m^{-1/2}H_mS_m^{-1/2},
    \qquad
    \mathbf c_m
    :=
    S_m^{-1/2}\mathbf g_m,
    \label{eq:orth_reduced_quantities}
\end{equation}
where $B_m$ is the orthonormal reduced operator and $\mathbf c_m$ is the
corresponding reduced right-hand side.

Introducing the orthonormal coordinate vector
\begin{equation}
    \mathbf y
    :=
    S_m^{1/2}\boldsymbol\gamma,
\end{equation}
the projected system becomes
\begin{equation}
    B_m\mathbf y
    =
    \mathbf c_m,
    \qquad
    x_m
    =
    Q_m\mathbf y.
    \label{eq:orth_reduced}
\end{equation}

Equations~\eqref{eq:qks_projected} and~\eqref{eq:orth_reduced} describe the
same QKS solution in two different coordinate systems.  The algorithm
itself solves $H_m\boldsymbol\gamma=\mathbf g_m$; the variables
$(B_m,\mathbf c_m,\mathbf y)$ are introduced only to expose the geometry
of the reduced problem and simplify the normalized-state sensitivity
analysis.

\subsection{Normalized-state error and directional sensitivity}

Because the target is a quantum state, global phase is physically
irrelevant.  For any two nonzero vectors $v$ and $w$, we therefore define
the phase-invariant angle
\begin{equation}
    \Theta(v,w)
    :=
    \arccos
    \frac{
    |\langle v,w\rangle|
    }{
    \|v\|_2\|w\|_2
    }
    \in[0,\pi/2].
    \label{eq:phase_invariant_angle}
\end{equation}
This angle measures the difference between the normalized directions of
$v$ and $w$.

To isolate perturbations that actually change the direction of the reduced
solution $\mathbf y$, define the orthogonal projector
\begin{equation}
    P_{\mathbf y}^{\perp}
    :=
    I-
    \frac{
    \mathbf y\mathbf y^\dagger
    }{
    \|\mathbf y\|_2^2
    }.
    \label{eq:transverse_projector}
\end{equation}
The projector $P_{\mathbf y}^{\perp}$ removes the component parallel to
$\mathbf y$ and retains only the transverse component that can change the
normalized state.

We then define
\begin{equation}
    K_B
    :=
    \|P_{\mathbf y}^{\perp}B_m^{-1}\|_2,
    \qquad
    K_c
    :=
    \frac{K_B}{\|\mathbf y\|_2}.
    \label{eq:sensitivity_factors}
\end{equation}
Here $K_B$ measures the directional sensitivity of the normalized solution
to perturbations of the reduced operator, while $K_c$ measures the
corresponding sensitivity to perturbations of the reduced right-hand side.

Finally, for a coefficient vector $\boldsymbol\gamma$, the LCU
reconstruction success probability is
\begin{equation}
    p_{\rm succ}
    :=
    \frac{
    \boldsymbol\gamma^\dagger
    S_m
    \boldsymbol\gamma
    }{
    \|\boldsymbol\gamma\|_1^2
    }.
    \label{eq:psucc_definition}
\end{equation}
Thus $p_{\rm succ}^{-1}$ quantifies the expected repetition overhead of
state reconstruction.

The quantities
\begin{equation}
    m,\qquad
    T,\qquad
    p_{\rm succ}^{-1},
    \qquad
    K_B,\qquad
    K_c
\end{equation}
are the condition-sensitive quantities that will organize the end-to-end
QKS complexity analysis below.  The first three describe the structural
cost of representing and reconstructing the solution, while $K_B$ and
$K_c$ describe the sensitivity of its normalized direction.

\begin{theorem}[General QKS complexity]
\label{thm:general_complexity}
For target normalized-state error $\epsilon$, the expected end-to-end QKS
complexity can be written schematically as
\begin{equation}
 C_{\rm QKS}
 =\widetilde O\!\left[
 C_H(T)\mathcal Q_H\!\left(
 m,\frac{\epsilon}{mK_B}
 \right)
 +C_g(T)\mathcal Q_g\!\left(
 m,\frac{\epsilon}{\sqrt m K_c}
 \right)
 +m^3
 +\frac{m+C_{\rm K}(T)}{p_{\rm succ}}
 \right],
 \label{eq:general_qks_complexity}
\end{equation}
where $T=(m-1)\tau$ and
\begin{equation}
 p_{\rm succ}
 =\frac{\boldsymbol\gamma^\dagger S_m\boldsymbol\gamma}
 {\|\boldsymbol\gamma\|_1^2}.
\end{equation}
Thus the condition-sensitive quantities separate into the structural factors
$m,T,p_{\rm succ}^{-1}$ and the normalized-state sensitivity factors
$K_B,K_c$. Implementation-dependent baseline costs are kept explicit and are
assumed not to introduce an additional inverse dependence on the vanishing
solution-relevant eigenvalue.
\end{theorem}

\paragraph{Interpretation.}
Theorem~\ref{thm:general_complexity} is the organizing decomposition for the
remainder of the paper.  The estimation terms are controlled by the
precisions required for the projected operator and projected right-hand side,
the classical reduced solve contributes the $m^3$ term, and LCU reconstruction
contributes the factor $p_{\rm succ}^{-1}$.  Thus the conditioning question
separates into two tasks: controlling the structural quantities
$m,T,p_{\rm succ}^{-1}$ and controlling the normalized-state sensitivity
quantities $K_B,K_c$.

In the following, Section~\ref{sec:structural} closes the structural
complexity channels by establishing compact representability and bounding
$m$, $T$, and $p_{\rm succ}^{-1}$.
Section~\ref{sec:directional} then studies the retained-small-eigenvalue
regime, establishes directional stability, transfers the compact subspace
approximation to the actual Galerkin output, and bounds the corresponding
sensitivity factors $K_B$ and $K_c$.
Section~\ref{sec:e2e_result} combines these ingredients into an end-to-end
complexity result that does not inherit the divergence of the global
condition number, and instantiates this result under several access models.
Finally, Section~\ref{sec:separation} places our result in the context of
existing beyond-condition-number methods by identifying both an overlap
regime and an explicit separation.

\section{Spectral Compression and Compact QKS Representation}
\label{sec:structural}

This section establishes the structural part of the beyond-conditioning
mechanism.  We ask whether the solution-relevant inverse response can be
captured inside a compact time-evolution Krylov subspace even when a
solution-relevant eigenvalue becomes arbitrarily small.  The objective is
to identify conditions under which the required Krylov dimension $m$, the
maximum evolution time $T$, and the inverse reconstruction success
probability $p_{\rm succ}^{-1}$ remain controlled without inheriting the
divergence of that small eigenvalue.

The argument proceeds in four steps.  We first identify the spectral modes
that carry the normalized solution and state the conditions under which
these modes are resolvable by a bounded-time QKS basis.  We then rewrite
QKS approximation as a weighted exponential approximation of the inverse
function.  Next, we construct localized time-evolution filters that isolate
the active spectral modes.  Finally, we combine these ingredients to prove
that the solution-relevant information can be compressed into a stable QKS
subspace with bounded structural resources.

This section establishes \emph{representability}: it proves that an
accurate solution vector exists inside $\mathcal K_m$ and that the resources
required to represent and reconstruct it need not diverge with the small
eigenvalue.  Whether the Galerkin solution actually returned by QKS
inherits this accuracy depends on the directional sensitivity of the
reduced solve and is addressed separately in
Sec.~\ref{sec:directional}.

\subsection{Solution-weighted spectral structure}
\label{subsec:solution_weighted_structure}

We begin by identifying which spectral modes actually matter to the target
solution.  Let
\begin{equation}
    A|u_\ell\rangle
    =
    \lambda_\ell |u_\ell\rangle,
    \qquad
    |b\rangle
    =
    \sum_{\ell=1}^{N}
    b_\ell |u_\ell\rangle,
\end{equation}
so that
\begin{equation}
    x
    :=
    A^{-1}|b\rangle
    =
    \sum_{\ell=1}^{N}
    \frac{b_\ell}{\lambda_\ell}
    |u_\ell\rangle.
\end{equation}

Define the normalized solution weight of eigenmode $\ell$ by
\begin{equation}
    w_\ell
    :=
    \frac{|b_\ell|^2/\lambda_\ell^2}
    {\sum_{k=1}^{N}|b_k|^2/\lambda_k^2},
    \qquad
    \sum_{\ell=1}^{N}w_\ell=1.
    \label{eq:solution_weight}
\end{equation}
Thus $w_\ell$ is the fraction of the squared solution norm carried by
$|u_\ell\rangle$.

For a prescribed solution-tail tolerance $\epsilon_{\rm sub}$, define the
\emph{active set}
\begin{equation}
    \mathcal A
    \in
    \operatorname*{arg\,min}_{\mathcal S\subseteq\{1,\ldots,N\}}
    \left\{
    |\mathcal S|:
    \sum_{\ell\in\mathcal S}w_\ell
    \ge
    1-\epsilon_{\rm sub}^2
    \right\},
    \label{eq:active_set_definition}
\end{equation}
and let
\begin{equation}
    r:=|\mathcal A|,
    \qquad
    \mathcal I
    :=
    \{1,\ldots,N\}\setminus\mathcal A.
    \label{eq:inactive_set_definition}
\end{equation}
The active set therefore contains the smallest number of spectral modes
needed to capture the target fraction of the normalized solution.

The two structural assumptions below play distinct roles.  Assumption
(A1) controls how many spectral modes carry the solution-relevant
information, whereas Assumption (A2) controls whether those modes can be
resolved by a stable bounded-time QKS basis.

\paragraph{(A1) Solution concentration.}
We assume that the number of active modes remains bounded with respect to
the vanishing solution-relevant eigenvalue:
\begin{equation}
    r=O(1).
    \label{eq:A1}
\end{equation}
By construction,
\begin{equation}
    \sum_{\ell\in\mathcal A}
    w_\ell
    \ge
    1-\epsilon_{\rm sub}^2.
    \label{eq:A1_weight}
\end{equation}
Equivalently, if
\begin{equation}
    x_{\mathcal I}
    :=
    \sum_{\ell\in\mathcal I}
    \frac{b_\ell}{\lambda_\ell}
    |u_\ell\rangle
\end{equation}
denotes the inactive contribution to the exact solution, then
\begin{equation}
    \boxed{
    \frac{\|x_{\mathcal I}\|_2}{\|x\|_2}
    \le
    \epsilon_{\rm sub}.}
    \label{eq:A1_tail}
\end{equation}

\paragraph{(A2) Bounded-time weighted spectral resolvability.}

For each eigenmode $\ell$, define its QKS time-evolution phase
\begin{equation}
    \theta_\ell
    :=
    \tau\lambda_\ell
    \pmod{2\pi},
\end{equation}
where $\tau>0$ is the time step used to generate the QKS basis.
Because the basis states contain the factors
$e^{-ij\tau\lambda_\ell}=e^{-ij\theta_\ell}$, spectral distinguishability
is determined by phase separation.  We measure this separation using
\begin{equation}
    d_{\mathbb T}(\theta,\theta')
    :=
    \min_{k\in\mathbb Z}
    |\theta-\theta'+2\pi k|.
\end{equation}

First, let $\delta_{\mathcal A}>0$ denote a uniform lower bound on the
pairwise phase separation of distinct active modes.  We assume
\begin{equation}
    d_{\mathbb T}(\theta_a,\theta_{a'})
    \ge
    \delta_{\mathcal A},
    \qquad
    a\neq a',
    \quad
    a,a'\in\mathcal A,
    \label{eq:A2_gap}
\end{equation}
where $\delta_{\mathcal A}$ does not shrink with the vanishing
solution-relevant eigenvalue.

Second, let $\beta_{\mathcal A}>0$ denote a uniform lower bound on the
input-state amplitude of every active mode.  We assume
\begin{equation}
    |b_a|
    \ge
    \beta_{\mathcal A},
    \qquad
    a\in\mathcal A,
    \label{eq:A2_overlap}
\end{equation}
with $\beta_{\mathcal A}$ likewise independent of the vanishing eigenvalue.

Next, choose a fixed resolution radius
\begin{equation}
    0<\rho<\frac{\delta_{\mathcal A}}{2}.
\end{equation}
An inactive mode is called \emph{unresolved} if its phase lies within
distance $\rho$ of at least one active phase.  Define
\begin{equation}
    \mathcal U_\rho
    :=
    \left\{
    \ell\in\mathcal I:
    \min_{a\in\mathcal A}
    d_{\mathbb T}(\theta_\ell,\theta_a)
    <
    \rho
    \right\},
    \label{eq:unresolved_inactive_set}
\end{equation}
and quantify its total input-state amplitude by
\begin{equation}
    \epsilon_{\rm U}
    :=
    \left(
    \sum_{\ell\in\mathcal U_\rho}
    |b_\ell|^2
    \right)^{1/2}.
    \label{eq:unresolved_weight}
\end{equation}
These unresolved modes cannot in general be suppressed by a bounded-time
filter while preserving the nearby active mode.  We therefore retain
$\epsilon_{\rm U}$ explicitly in the approximation error and assume that it
does not grow as the solution-relevant eigenvalue vanishes.

Finally, the nonorthogonal time-evolution basis itself must remain
nondegenerate.  For the Krylov dimension $m$ under consideration, let
\begin{equation}
    \mathcal R
    \subseteq
    \{1,\ldots,N\},
    \qquad
    |\mathcal R|\ge m,
\end{equation}
be a spectral support set used to certify basis stability.  Define
\begin{equation}
    (V_{\mathcal R})_{\ell j}
    :=
    e^{-ij\tau\lambda_\ell},
    \qquad
    \ell\in\mathcal R,
    \quad
    j=0,\ldots,m-1,
\end{equation}
and
\begin{equation}
    W_{\mathcal R}
    :=
    \operatorname{diag}
    \left(
    |b_\ell|^2
    \right)_{\ell\in\mathcal R}.
\end{equation}
We assume
\begin{equation}
    \boxed{
    \sigma_{\min}
    \!\left(
    W_{\mathcal R}^{1/2}V_{\mathcal R}
    \right)
    \ge
    s_0>0,}
    \label{eq:A2_basis}
\end{equation}
where $s_0$ is independent of the vanishing eigenvalue.  The time step is
also assumed bounded,
\begin{equation}
    \tau=O(1),
    \label{eq:A2_tau}
\end{equation}
with respect to that eigenvalue.

Assumptions (A1)--(A2) therefore specify a solution-weighted spectral
regime in which only a compact set of modes is relevant, those modes remain
resolvable at bounded evolution time, and the resulting Krylov basis
remains stable.  These assumptions alone do not yet prove that
$A^{-1}|b\rangle$ can be represented accurately in $\mathcal K_m$.
We next convert this question into a weighted approximation problem.

\subsection{Weighted time-evolution representation}
\label{subsec:weighted_representation}

The first step is to express the QKS approximation error directly in the
eigenbasis of $A$.

\begin{proposition}[Weighted exponential representation]
\label{prop:weighted_exp}

For any coefficient vector
$\boldsymbol\gamma\in\mathbb C^m$,
\begin{equation}
    \boxed{
    \|x-\Phi_m\boldsymbol\gamma\|_2^2
    =
    \sum_{\ell=1}^{N}
    |b_\ell|^2
    \left|
    \frac1{\lambda_\ell}
    -
    \sum_{j=0}^{m-1}
    \gamma_j e^{-ij\tau\lambda_\ell}
    \right|^2.}
    \label{eq:weighted_exp_identity}
\end{equation}
\end{proposition}

Equation~\eqref{eq:weighted_exp_identity} shows that constructing an
accurate QKS vector is equivalent to approximating the inverse function
$1/\lambda$ by a finite linear combination of time-evolution exponentials,
with the approximation weighted by the input amplitudes $|b_\ell|^2$.

\paragraph{Proof sketch.}
Expand both the exact solution
$x=A^{-1}|b\rangle$ and the QKS vector
$\Phi_m\boldsymbol\gamma$ in the eigenbasis of $A$.
The coefficient of eigenmode $\ell$ differs by
\begin{equation}
    b_\ell
    \left(
    \frac1{\lambda_\ell}
    -
    \sum_{j=0}^{m-1}
    \gamma_j e^{-ij\tau\lambda_\ell}
    \right).
\end{equation}
Orthogonality of the eigenvectors then gives
Eq.~\eqref{eq:weighted_exp_identity}.
The complete proof is given in
Appendix~\ref{app:weighted_exp_proof}.

The proposition reduces the structural QKS problem to finite exponential
interpolation on the solution-relevant spectrum.  We next construct the
localized filters used to perform this interpolation.

\subsection{Localized time-evolution filters}
\label{subsec:localized_filters}

For each active eigenmode, we introduce a bounded-time filter that preserves
that mode while suppressing modes that are sufficiently separated from it
in time-evolution phase.

\begin{lemma}[Localized time-evolution filters]
\label{lem:localized_filter}

For each $a\in\mathcal A$, define
\begin{equation}
    F_{a,m}(\lambda)
    :=
    \frac1m
    \sum_{j=0}^{m-1}
    e^{-ij(\tau\lambda-\theta_a)}.
    \label{eq:filter_definition}
\end{equation}
At the target active eigenvalue,
\begin{equation}
    F_{a,m}(\lambda_a)=1.
    \label{eq:filter_target}
\end{equation}
For any spectral value $\lambda$ satisfying
\begin{equation}
    d_{\mathbb T}(\tau\lambda,\theta_a)
    \ge
    \delta>0,
\end{equation}
the filter obeys
\begin{equation}
    \boxed{
    |F_{a,m}(\lambda)|
    \le
    \frac{1}{m\sin(\delta/2)}.}
    \label{eq:filter_bound}
\end{equation}
\end{lemma}

Thus the target phase adds coherently and is reproduced exactly, whereas a
phase-separated mode is suppressed as $O(1/m)$ for fixed separation.

\paragraph{Proof sketch.}
For a principal phase difference
\begin{equation}
    \phi
    :=
    \tau\lambda-\theta_a
    \pmod{2\pi},
\end{equation}
the filter is a finite geometric series,
\begin{equation}
    F_{a,m}(\lambda)
    =
    \frac{1-e^{-im\phi}}
    {m(1-e^{-i\phi})}.
\end{equation}
The numerator is uniformly bounded, while
\begin{equation}
    |1-e^{-i\phi}|
    =
    2|\sin(\phi/2)|.
\end{equation}
A phase separation of at least $\delta$ therefore gives
Eq.~\eqref{eq:filter_bound}.
The complete proof is given in
Appendix~\ref{app:localized_filter_proof}.

A single filter isolates one active phase.  The remaining task is to combine
these filters so that $1/\lambda$ is reproduced simultaneously across the
entire active spectrum while controlling the inactive contribution.  This
gives the main structural result.

\subsection{Spectral compression of QKS structural resources}
\label{subsec:structural_resources}

Before stating the theorem, recall that the QKS reconstruction success
probability for a coefficient vector $\boldsymbol\gamma$ is
\begin{equation}
    p_{\rm succ}
    :=
    \frac{
    \boldsymbol\gamma^\dagger
    S_m
    \boldsymbol\gamma
    }{
    \|\boldsymbol\gamma\|_1^2
    }.
    \label{eq:structural_psucc_definition}
\end{equation}
The maximum evolution time used by the $m$ time-evolution basis states is
\begin{equation}
    T=(m-1)\tau.
    \label{eq:structural_T_definition}
\end{equation}

The following theorem combines solution concentration, phase resolvability,
and Gram stability to show that the solution-relevant inverse response can
be compressed into a compact and reconstructible QKS subspace.

\begin{theorem}[Resolvable spectral compression of QKS resources]
\label{thm:structural}

There exist universal constants
$C_{\rm U},C_{\rm res},C_{\rm app}>0$
such that the following holds.

Assume (A1)--(A2), and let $\epsilon_K$ be the target relative
subspace-approximation error.  Define the remaining approximation margin
\begin{equation}
    \Delta_K
    :=
    \epsilon_K
    -
    \epsilon_{\rm sub}
    -
    C_{\rm U}
    \frac{r}{\beta_{\mathcal A}}
    \epsilon_{\rm U}.
    \label{eq:DeltaK_definition}
\end{equation}
Suppose
\begin{equation}
    \Delta_K>0.
    \label{eq:DeltaK_positive}
\end{equation}
If the Krylov dimension satisfies
\begin{equation}
    \boxed{
    m
    \ge
    \max\left\{
    C_{\rm res}
    \frac{r}{
    \sin(\delta_{\mathcal A}/2)},
    \;
    C_{\rm app}
    \frac{r}{
    \beta_{\mathcal A}
    \Delta_K
    \sin(\rho/2)}
    \right\},
    }
    \label{eq:structural_m_requirement}
\end{equation}
then there exists a vector
\begin{equation}
    z_m\in\mathcal K_m
\end{equation}
such that
\begin{equation}
    \boxed{
    \frac{\|x-z_m\|_2}{\|x\|_2}
    \le
    \epsilon_K.}
    \label{eq:structural_target_error}
\end{equation}

Moreover, the QKS Gram matrix satisfies
\begin{equation}
    \boxed{
    \lambda_{\min}(S_m)
    \ge
    s_0^2,}
    \label{eq:structural_gram_bound}
\end{equation}
and therefore the reconstruction probability satisfies
\begin{equation}
    \boxed{
    p_{\rm succ}
    \ge
    \frac{s_0^2}{m}.}
    \label{eq:structural_psucc_bound}
\end{equation}

Consequently,
\begin{equation}
    m
    =
    O\!\left(
    \max\left\{
    \frac{r}{
    \sin(\delta_{\mathcal A}/2)},
    \;
    \frac{r}{
    \beta_{\mathcal A}
    \Delta_K
    \sin(\rho/2)}
    \right\}
    \right),
    \label{eq:structural_m_scaling}
\end{equation}
with
\begin{equation}
    T=(m-1)\tau,
    \qquad
    p_{\rm succ}^{-1}
    \le
    \frac{m}{s_0^2}.
    \label{eq:structural_resource_scaling}
\end{equation}

For fixed target accuracy, if
\begin{equation}
    r=O(1),
    \qquad
    \delta_{\mathcal A},\rho,
    \beta_{\mathcal A},s_0,
    \Delta_K=\Omega(1),
    \qquad
    \tau=O(1),
\end{equation}
with respect to a vanishing solution-relevant eigenvalue, then
\begin{equation}
    \boxed{
    m=O(1),
    \qquad
    T=O(1),
    \qquad
    p_{\rm succ}^{-1}=O(1)}
    \label{eq:structural_beyond_small_eigenvalue}
\end{equation}
with respect to that eigenvalue.
Thus the retained small eigenvalue does not force the structural QKS
resources to diverge.
\end{theorem}

\paragraph{Proof sketch.}
For each active mode,
Lemma~\ref{lem:localized_filter} provides a time-evolution filter that has
unit response on that mode and suppresses phase-separated modes.  Evaluating
these filters on the active spectrum produces an interpolation matrix whose
diagonal entries are one and whose off-diagonal entries decay as
$1/[m\sin(\delta_{\mathcal A}/2)]$.  The first requirement in
Eq.~\eqref{eq:structural_m_requirement} therefore makes this interpolation
matrix uniformly invertible.

Stable inversion of the active interpolation matrix yields a finite
exponential function
\begin{equation}
    q_m(\lambda)
\end{equation}
satisfying
\begin{equation}
    q_m(\lambda_a)
    =
    \frac1{\lambda_a},
    \qquad
    a\in\mathcal A.
\end{equation}
Hence
\begin{equation}
    z_m
    :=
    q_m(A)|b\rangle
    \in
    \mathcal K_m
\end{equation}
reproduces the exact inverse response on the entire active subspace.

The remaining approximation error has three sources.  The exact solution
weight outside the active set contributes at most
$\epsilon_{\rm sub}$; unresolved inactive modes contribute proportionally
to $(r/\beta_{\mathcal A})\epsilon_{\rm U}$; and resolved inactive modes
are suppressed by the localized filters as
\begin{equation}
    O\!\left(
    \frac{r}{
    \beta_{\mathcal A}
    m\sin(\rho/2)}
    \right).
\end{equation}
The second requirement in
Eq.~\eqref{eq:structural_m_requirement} places the last contribution inside
the remaining error budget $\Delta_K$, proving
Eq.~\eqref{eq:structural_target_error}.

Finally, the weighted phase-matrix condition in (A2) implies
\begin{equation}
    S_m
    \succeq
    V_{\mathcal R}^{\dagger}
    W_{\mathcal R}
    V_{\mathcal R},
\end{equation}
and therefore
\begin{equation}
    \lambda_{\min}(S_m)
    \ge
    s_0^2.
\end{equation}
Using
\begin{equation}
    \|\boldsymbol\gamma\|_1^2
    \le
    m\|\boldsymbol\gamma\|_2^2
\end{equation}
in Eq.~\eqref{eq:structural_psucc_definition} then gives
\begin{equation}
    p_{\rm succ}
    \ge
    \frac{s_0^2}{m}.
\end{equation}
The complete proof, including explicit numerical constants, is given in
Appendix~\ref{app:structural_proof}.

The theorem establishes the structural compression mechanism required for
the beyond-conditioning regime.  Up to the prescribed approximation error,
the solution-relevant inverse response is contained in a QKS subspace whose
dimension does not grow with the vanishing solution-relevant eigenvalue.
The same Gram stability that keeps this representation nondegenerate also
controls the reconstruction probability, while bounded $\tau$ converts the
dimension bound directly into a bounded maximum evolution time.  Hence
\begin{equation}
    \boxed{
    m,\qquad
    T,\qquad
    p_{\rm succ}^{-1}
    }
\end{equation}
do not inherit the divergence of the retained small eigenvalue.

What remains unresolved is not the existence or structural cost of the
compact representation, but whether the Galerkin solve actually returned
by QKS inherits its approximation accuracy.  That question depends on the
directional sensitivity of the reduced system and is addressed in
Sec.~\ref{sec:directional}.
\section{Directional Beyond-Conditioning with a Retained Small Eigenvalue}
\label{sec:directional}

The preceding section establishes that the target inverse response can be
represented accurately within a compact and stable QKS subspace, with the
structural resources controlled independently of a vanishing
solution-relevant eigenvalue.  This section addresses the remaining
question: whether the \emph{actual QKS solution} remains stable when the
reduced problem itself retains that small eigenvalue and becomes
ill conditioned.

The argument proceeds in six stages.  We first introduce an equivalent
orthonormal representation of the QKS projected system and show that the
change of coordinates from the projected quantities
$(H_m,\mathbf g_m)$ to $(B_m,\mathbf c_m)$ is stable when the theoretical
Gram geometry $S_m$ is nondegenerate.  We then specify
the retained-soft-mode regime.  Third, we prove that its exact
normalized-state sensitivity remains bounded even though
$\|B_m^{-1}\|_2$ diverges.  Fourth, we show that this directional
protection persists under finite perturbations of the reduced system.
Fifth, we use the bounded directional sensitivity to transfer the compact
subspace approximation of the preceding section to the Galerkin solution
actually returned by QKS.  Finally, we translate the finite-stability
requirements back to the precision required for the measured projected
quantities.

\subsection{Equivalent orthonormal reduced coordinates and stable transfer}
\label{subsec:orthonormal_coordinates}

The QKS algorithm is constructed from the generally nonorthogonal
time-evolution basis $\Phi_m$ and solves
\begin{equation}
    H_m\boldsymbol\gamma=\mathbf g_m,
    \qquad
    S_m=\Phi_m^\dagger\Phi_m,
\end{equation}
with physical output
\begin{equation}
    x_m=\Phi_m\boldsymbol\gamma.
\end{equation}
For the sensitivity analysis below, it is convenient to express the same
projected problem in orthonormal coordinates.  Define
\begin{equation}
    Q_m:=\Phi_mS_m^{-1/2},
    \qquad
    B_m:=S_m^{-1/2}H_mS_m^{-1/2},
    \qquad
    \mathbf c_m:=S_m^{-1/2}\mathbf g_m,
\end{equation}
and
\begin{equation}
    \mathbf y:=S_m^{1/2}\boldsymbol\gamma.
\end{equation}
Then
\begin{equation}
    B_m\mathbf y=\mathbf c_m,
    \qquad
    x_m=\Phi_m\boldsymbol\gamma=Q_m\mathbf y.
    \label{eq:equivalent_reduced_representation}
\end{equation}
Thus $(B_m,\mathbf c_m)$ do not define a different algorithm or a different
physical output; they are an equivalent orthonormal coordinate
representation of the same QKS projected problem.

The next result shows that this coordinate transformation is stable whenever
the Krylov Gram matrix is stable.

\begin{theorem}[Stable transfer to orthonormal reduced coordinates]
\label{thm:SHg_Bc}

Assume
\begin{equation}
    \lambda_{\min}(S_m)\ge\sigma_S>0,
\end{equation}
where $S_m$ is the theoretical Gram matrix of the QKS basis.  Let
\begin{equation}
    \widehat H_m=H_m+\Delta H_m,
    \qquad
    \widehat{\mathbf g}_m=\mathbf g_m+\Delta\mathbf g_m
\end{equation}
denote finite errors in the projected quantities used by QKS, and define
their orthonormal-coordinate counterparts using the same exact Gram
geometry:
\begin{equation}
    \widehat B_m
    :=
    S_m^{-1/2}\widehat H_mS_m^{-1/2},
    \qquad
    \widehat{\mathbf c}_m
    :=
    S_m^{-1/2}\widehat{\mathbf g}_m.
\end{equation}
Then
\begin{equation}
    \boxed{
    \|\widehat B_m-B_m\|_2
    \le
    \frac{\|\Delta H_m\|_2}{\sigma_S},
    }
    \label{eq:B_SH_bound}
\end{equation}
and
\begin{equation}
    \boxed{
    \|\widehat{\mathbf c}_m-\mathbf c_m\|_2
    \le
    \frac{\|\Delta\mathbf g_m\|_2}{\sqrt{\sigma_S}}.
    }
    \label{eq:c_Sg_bound}
\end{equation}
Thus the theoretical change of coordinates introduces no inverse dependence
on a vanishing eigenvalue of $B_m$.
\end{theorem}

\paragraph{Proof sketch.}
Because $S_m$ is used only as a fixed theoretical coordinate transform,
\begin{equation}
    \widehat B_m-B_m
    =
    S_m^{-1/2}\Delta H_mS_m^{-1/2},
\end{equation}
and
\begin{equation}
    \widehat{\mathbf c}_m-\mathbf c_m
    =
    S_m^{-1/2}\Delta\mathbf g_m.
\end{equation}
The Gram lower bound gives
$\|S_m^{-1/2}\|_2\le\sigma_S^{-1/2}$, which directly yields
Eqs.~\eqref{eq:B_SH_bound}--\eqref{eq:c_Sg_bound}.
The complete proof is given in
Appendix~\ref{app:orthonormalization_proof}.

Having established that errors in the projected quantities
$(H_m,\mathbf g_m)$ transfer stably to the analysis coordinates
$(B_m,\mathbf c_m)$, we now analyze the intrinsic conditioning of this
equivalent reduced system.

\subsection{Solution-aligned retained ill-conditioning}
\label{subsec:directional_regime}

We now specify the reduced-system regime studied in the remainder of this
section.  Let
\begin{equation}
    B_m\mathbf v_i
    =
    \mu_i\mathbf v_i,
    \qquad
    0<\mu_1\le\mu_2\le\cdots,
\end{equation}
and decompose the reduced right-hand side as
\begin{equation}
    \mathbf c_m
    =
    c_1\mathbf v_1+\mathbf c_\perp,
    \qquad
    \mathbf c_\perp\perp\mathbf v_1.
\end{equation}

\paragraph{(B1) Single retained soft direction.}
The smallest reduced eigenvalue vanishes while the remaining reduced
spectrum stays separated from zero:
\begin{equation}
    \mu_1\rightarrow0,
    \qquad
    \mu_2\ge\Delta_B>0,
    \label{eq:B1}
\end{equation}
where $\Delta_B$ is independent of $\mu_1$.

\paragraph{(B2) Persistent solution relevance.}
The reduced right-hand side retains nonvanishing overlap with the soft
direction:
\begin{equation}
    |c_1|
    \ge
    \beta_B>0,
    \label{eq:B2}
\end{equation}
where $\beta_B$ is independent of $\mu_1$.

Thus the small reduced eigenvalue is neither truncated nor irrelevant to
the solution.  Under these assumptions,
$\|B_m^{-1}\|_2=\mu_1^{-1}$ diverges.  The next subsection asks whether
this divergence also makes the \emph{normalized solution direction}
singular.

\subsection{Ideal normalized-state directional stability}
\label{subsec:ideal_directional_stability}

For the exact reduced system
\begin{equation}
    B_m\mathbf y=\mathbf c_m,
\end{equation}
define
\begin{equation}
    P_{\mathbf y}^{\perp}
    :=
    I-
    \frac{\mathbf y\mathbf y^\dagger}{\|\mathbf y\|_2^2},
\end{equation}
and recall the normalized directional sensitivity factors
\begin{equation}
    K_B
    :=
    \|P_{\mathbf y}^{\perp}B_m^{-1}\|_2,
    \qquad
    K_c
    :=
    \frac{K_B}{\|\mathbf y\|_2}.
\end{equation}
The projector $P_{\mathbf y}^{\perp}$ removes amplification parallel to the
solution and therefore isolates the part of the inverse response that can
change its normalized direction.

\begin{theorem}[Directional stability under solution-aligned ill-conditioning]
\label{thm:solution_aligned}

Under (B1)--(B2),
\begin{equation}
    \boxed{
    K_B
    \le
    \frac1{\mu_2}
    \left(
    1+
    \frac{\|\mathbf c_\perp\|_2}{|c_1|}
    \right)
    =
    O(1),
    }
    \label{eq:KB_bound}
\end{equation}
and
\begin{equation}
    \boxed{
    K_c
    \le
    \frac{\mu_1}{|c_1|\mu_2}
    \left(
    1+
    \frac{\|\mathbf c_\perp\|_2}{|c_1|}
    \right)
    =
    O(\mu_1).
    }
    \label{eq:Kc_bound}
\end{equation}
Hence the normalized-state sensitivity remains bounded even though
\begin{equation}
    \|B_m^{-1}\|_2
    =
    \mu_1^{-1}
    \longrightarrow\infty.
\end{equation}
\end{theorem}

\paragraph{Proof sketch.}
The exact reduced solution has the spectral expansion
\begin{equation}
    \mathbf y
    =
    \frac{c_1}{\mu_1}\mathbf v_1
    +
    \sum_{i\ge2}
    \frac{c_i}{\mu_i}\mathbf v_i.
\end{equation}
Because $\mu_1\to0$, $\mu_2\ge\Delta_B$, and
$|c_1|\ge\beta_B$, the first term increasingly dominates the normalized
solution, so the angle between $\mathbf y$ and $\mathbf v_1$ decreases
proportionally to $\mu_1/\mu_2$.  The only divergent component of
$B_m^{-1}$ is also along $\mathbf v_1$.  Projection by
$P_{\mathbf y}^{\perp}$ therefore suppresses this increasingly
solution-parallel contribution, cancelling the apparent $1/\mu_1$
divergence and yielding Eq.~\eqref{eq:KB_bound}.  Since
$\|\mathbf y\|_2$ grows at least as $|c_1|/\mu_1$,
Eq.~\eqref{eq:Kc_bound} follows.
The complete proof is given in
Appendix~\ref{app:solution_aligned_proof}.

The theorem establishes the intrinsic directional protection of the
\emph{exact} reduced system.  An implementable algorithm, however, uses
finite-precision estimates.  We therefore next ask whether the same
geometric protection survives finite perturbations of $(B_m,\mathbf c_m)$.

\subsection{Finite-error robustness of the directional mechanism}
\label{subsec:finite_directional_stability}

Let
\begin{equation}
    \widehat B_m
    =
    B_m+E,
    \qquad
    \widehat{\mathbf c}_m
    =
    \mathbf c_m+\mathbf f
\end{equation}
denote a perturbed reduced system, where $E=E^\dagger$, with
\begin{equation}
    \|E\|_2\le\eta_B,
    \qquad
    \|\mathbf f\|_2\le\eta_c.
\end{equation}
Here $\eta_B$ and $\eta_c$ quantify the finite errors in the reduced
operator and reduced right-hand side, respectively.

\begin{theorem}[Finite-error directional stability]
\label{thm:finite_stability}

Assume (B1)--(B2).  Since
$\mathbf c_m=Q_m^\dagger|b\rangle$ with $Q_m$ an isometry and
$\|b\|_2=1$, one automatically has $\|\mathbf c_m\|_2\le1$.
There exist fixed positive constants
$c_{\rm gap}$, $c_{\rm ov}$, $C_{\rm rot}$, and $C_{\rm stab}$,
independent of $\mu_1$, such that if
\begin{equation}
    \mu_1
    \le
    c_{\rm gap}\Delta_B,
    \qquad
    \eta_B
    \le
    c_{\rm gap}\Delta_B,
    \label{eq:finite_gap_conditions}
\end{equation}
and
\begin{equation}
    \eta_c
    +
    C_{\rm rot}\frac{\eta_B}{\Delta_B}
    \le
    c_{\rm ov}\beta_B,
    \label{eq:finite_overlap_condition}
\end{equation}
and assume that $\widehat B_m$ is nonsingular.  The exactly singular
finite-precision case is treated separately in
Appendix~\ref{app:near_singular_regularization}.  Then the perturbed
system retains an isolated soft eigendirection
$\widehat{\mathbf v}_1$ satisfying
\begin{equation}
    |\langle
    \widehat{\mathbf v}_1,
    \widehat{\mathbf c}_m
    \rangle|
    =
    \Omega(\beta_B).
\end{equation}
Moreover, for
\begin{equation}
    \widehat{\mathbf y}
    :=
    \widehat B_m^{-1}
    \widehat{\mathbf c}_m,
\end{equation}
the normalized-direction error satisfies
\begin{equation}
    \boxed{
    \Theta(\mathbf y,\widehat{\mathbf y})
    \le
    C_{\rm stab}
    \left[
    \frac{\mu_1}{\beta_B\Delta_B}
    +
    \frac{\eta_B}{\Delta_B}
    +
    \frac{(1+\eta_c)(\mu_1+\eta_B)}
    {\beta_B\Delta_B}
    \right],
    }
    \label{eq:finite_state_bound}
\end{equation}
whenever the right-hand side lies in the perturbative regime.
In particular, the bound contains no inverse dependence on $\mu_1$.
\end{theorem}

\paragraph{Proof sketch.}
The proof first uses the fixed reduced spectral gap to show that the soft
eigendirection rotates by at most
$O(\eta_B/\Delta_B)$ under the perturbation $E$.  Condition
Eq.~\eqref{eq:finite_overlap_condition} is a sufficient precision
requirement ensuring that the combined operator and right-hand-side errors
do not destroy the nonvanishing overlap with this perturbed soft direction.
The remaining perturbed eigenvalues stay separated from zero, so the
transverse part of $\widehat{\mathbf y}$ remains bounded, whereas its soft
component is amplified by the isolated perturbed small eigenvalue.
Consequently $\widehat{\mathbf y}$ remains aligned with
$\widehat{\mathbf v}_1$.

The final bound follows by combining the three angles
\begin{equation}
    \mathbf y
    \longleftrightarrow
    \mathbf v_1
    \longleftrightarrow
    \widehat{\mathbf v}_1
    \longleftrightarrow
    \widehat{\mathbf y}.
\end{equation}
Importantly, the allowed perturbation sizes are controlled by
$\Delta_B$ and $\beta_B$, rather than by the vanishing eigenvalue
$\mu_1$ itself.  The complete proof, including explicit numerical
constants, is given in
Appendix~\ref{app:finite_stability_proof}.

\begin{remark}[Near-singular estimated soft modes]
\label{rem:regularized_estimated_solve}
The theorem does not require the estimated soft eigenvalue to remain
bounded away from zero.  If finite precision produces an exactly singular
estimate, the isolated estimated eigenvalue may be replaced by a small
spectral floor of order the target state accuracy.  This regularization is
used only to define the finite-precision reduced solve and does not require
a cutoff proportional to the physical eigenvalue $\mu_1$.
Details are given in
Appendix~\ref{app:near_singular_regularization}.
\end{remark}

We have now established that the directional protection is stable both
locally and under finite reduced-system errors.  The remaining question is
whether this stability is sufficient to ensure that QKS actually returns
the accurate vector whose existence was established by the structural
analysis.

\subsection{From compact subspace representation to the QKS output}
\label{subsec:galerkin_transfer}

The structural theorem guarantees that the Krylov subspace
$\mathcal K_m$ contains an accurate approximation to the exact solution.
QKS, however, does not explicitly construct that particular vector; it
returns the Galerkin solution $x_m$ obtained from the projected system.
The next theorem shows that bounded directional sensitivity is precisely
what transfers the subspace approximation accuracy to the actual QKS
output.

\begin{theorem}[Transfer from subspace approximation to QKS output]
\label{thm:galerkin_transfer}

Let
\begin{equation}
    x=A^{-1}|b\rangle
\end{equation}
be the exact unnormalized solution, and let $x_m\in\mathcal K_m$ denote
the Galerkin solution returned by QKS.  Suppose the Krylov subspace contains
a vector $z\in\mathcal K_m$ satisfying
\begin{equation}
    \frac{\|x-z\|_2}{\|x\|_2}
    \le
    \epsilon_K
    <1.
    \label{eq:good_candidate}
\end{equation}
Then
\begin{equation}
    \boxed{
    \sin\Theta(x_m,z)
    \le
    \frac{K_B\epsilon_K}{1-\epsilon_K}.
    }
    \label{eq:galerkin_candidate_angle}
\end{equation}
Consequently,
\begin{equation}
    \boxed{
    \Theta(x_m,x)
    \le
    \arcsin\!\left(
    \frac{K_B\epsilon_K}{1-\epsilon_K}
    \right)
    +
    \arcsin(\epsilon_K),
    }
    \label{eq:galerkin_transfer_bound}
\end{equation}
whenever the first inverse-sine argument is at most one.

Under the retained-soft-mode conditions (B1)--(B2),
Theorem~\ref{thm:solution_aligned} gives $K_B=O(1)$, and therefore
\begin{equation}
    \boxed{
    \Theta(x_m,x)
    =
    O(\epsilon_K)
    \qquad
    (\epsilon_K\rightarrow0).
    }
    \label{eq:galerkin_transfer_scaling}
\end{equation}
Thus the compact solution representation established in the preceding
section is inherited by the actual QKS output.
\end{theorem}

\paragraph{Proof sketch.}
Both the accurate candidate $z$ and the Galerkin solution $x_m$ lie in the
same Krylov subspace.  Expressing them temporarily in the equivalent
orthonormal coordinates introduced in
Sec.~\ref{subsec:orthonormal_coordinates}, the Galerkin equations show that
their difference is generated by the residual $x-z$ through $B_m^{-1}$.
Only the component transverse to the Galerkin solution changes the
normalized direction, and this component is bounded by
$K_B\|x-z\|_2$.  The lower bound
$\|z\|_2\ge(1-\epsilon_K)\|x\|_2$
then gives Eq.~\eqref{eq:galerkin_candidate_angle}.  Combining the angle
between $x_m$ and $z$ with the angle between $z$ and $x$ yields
Eq.~\eqref{eq:galerkin_transfer_bound}.
The complete proof is given in
Appendix~\ref{app:galerkin_transfer_proof}.

At this point the theoretical chain is complete in the orthonormal reduced
coordinates: the subspace contains an accurate solution representation,
the reduced directional sensitivity remains bounded, and the Galerkin
solve transfers that representation accuracy to the actual QKS output.
It remains only to translate the finite-stability requirements back to the
quantities measured by the quantum procedure.

\subsection{Projected-measurement precision}
\label{subsec:measurement_precision}

Combining the stable coordinate transformation of
Theorem~\ref{thm:SHg_Bc} with the finite-error directional stability of
Theorem~\ref{thm:finite_stability} yields the precision requirement for the
actual projected quantities.

\begin{corollary}[Projected-measurement precision is independent of the retained small eigenvalue]
\label{cor:no_inverse_precision}

Define the effective reduced-system errors
\begin{equation}
    \eta_B^{\rm eff}
    :=
    \frac{\|\Delta H_m\|_2}{\sigma_S},
    \label{eq:etaB_eff}
\end{equation}
and
\begin{equation}
    \eta_c^{\rm eff}
    :=
    \frac{\|\Delta\mathbf g_m\|_2}{\sqrt{\sigma_S}}.
    \label{eq:etac_eff}
\end{equation}

Assume the hypotheses of
Theorems~\ref{thm:SHg_Bc} and~\ref{thm:finite_stability}, and suppose
$\eta_B^{\rm eff}$ and $\eta_c^{\rm eff}$ satisfy the finite-stability
requirements in
Eqs.~\eqref{eq:finite_gap_conditions}--\eqref{eq:finite_overlap_condition}.
Then the normalized reduced-state error obeys
Eq.~\eqref{eq:finite_state_bound} with
\begin{equation}
    \eta_B=\eta_B^{\rm eff},
    \qquad
    \eta_c=\eta_c^{\rm eff}.
\end{equation}

In particular, if
\begin{equation}
    \sigma_S,\Delta_B,\beta_B=\Omega(1),
\end{equation}
then the additive error tolerances required for the actually estimated
projected quantities $H_m$ and $\mathbf g_m$ can be chosen independently of
the vanishing eigenvalue $\mu_1$.  Consequently, projected-system
estimation does not acquire an inverse-$\mu_1$ cost through its precision
requirement.
\end{corollary}

\paragraph{Proof sketch.}
Theorem~\ref{thm:SHg_Bc} converts the physical estimation errors
$(\Delta H_m,\Delta\mathbf g_m)$ into the effective reduced perturbations
$(E,\mathbf f)$ with norms bounded by
$\eta_B^{\rm eff}$ and $\eta_c^{\rm eff}$.  Theorem~\ref{thm:finite_stability}
then converts these reduced perturbations into a normalized-state error.
Because neither step introduces an inverse dependence on $\mu_1$, the
required measurement precision need not become more stringent as the
retained eigenvalue vanishes.  Full details follow directly from
Appendices~\ref{app:orthonormalization_proof}
and~\ref{app:finite_stability_proof}.

\section{End-to-End QKS Complexity Beyond Global Conditioning}
\label{sec:e2e_result}

The preceding two sections establish the two ingredients needed for an
end-to-end beyond-conditioning result.  Section~\ref{sec:structural} shows
that the solution-relevant inverse response can be compressed into a QKS
subspace with bounded structural resources,
\begin{equation}
    m,\qquad T,\qquad p_{\rm succ}^{-1},
\end{equation}
without inheriting the divergence of the retained small eigenvalue.
Section~\ref{sec:directional} then shows that the reduced solve remains
stable in normalized direction, that this stability persists under finite
projected-system errors, and that the Galerkin output returned by QKS
inherits the accuracy of the compact subspace representation.

This section combines these results into the end-to-end complexity
statement.  The argument proceeds in three steps.  We first show that all
condition-sensitive quantities appearing in the general QKS complexity
bound remain controlled in the retained-small-eigenvalue regime.  We then
separate this intrinsic beyond-conditioning mechanism from the remaining
cost of physically estimating the projected system.  Finally, we
instantiate the latter cost under several representative matrix-access
models.


\subsection{Combined beyond-conditioning result}
\label{subsec:combined_result}

We first combine the structural and directional results without committing
to a particular implementation of the projected-system estimator.  This
isolates the intrinsic effect of the retained small eigenvalue from
access-dependent costs.

\begin{theorem}[End-to-end QKS result with a retained small eigenvalue]
\label{thm:combined}

Assume the structural conditions (A1)--(A2) of
Sec.~\ref{sec:structural}, and choose the Krylov dimension $m$ according to
Theorem~\ref{thm:structural} so that
\begin{equation}
    \frac{\|x-z_m\|_2}{\|x\|_2}
    \le
    \epsilon_K
\end{equation}
for some $z_m\in\mathcal K_m$.

Assume further that the equivalent orthonormal reduced system satisfies
the solution-aligned retained-soft-mode conditions (B1)--(B2) of
Sec.~\ref{sec:directional},
\begin{equation}
    \mu_1\rightarrow0,
    \qquad
    \mu_2\ge\Delta_B>0,
    \qquad
    |c_1|\ge\beta_B>0,
\end{equation}
and that the finite projected-system errors satisfy the stability
conditions of Theorem~\ref{thm:finite_stability}.

Then, for fixed target accuracy and fixed solution-weighted structural
parameters,
\begin{equation}
    \boxed{
    m=O(1),\qquad
    T=O(1),\qquad
    p_{\rm succ}^{-1}=O(1)}
    \label{eq:e2e_structural_resources}
\end{equation}
with respect to $\mu_1$, while
\begin{equation}
    \boxed{
    K_B=O(1),
    \qquad
    K_c=O(\mu_1).}
    \label{eq:e2e_directional_resources}
\end{equation}

Moreover, the QKS Galerkin output $x_m$ satisfies
\begin{equation}
    \Theta(x_m,x)
    =
    O(\epsilon_K),
    \label{eq:e2e_galerkin_accuracy}
\end{equation}
and the additive precision required for the measured projected quantities
$(H_m,\mathbf g_m)$ can be chosen without inverse dependence on
$\mu_1$.

Consequently, provided that the underlying matrix-access and
expectation-estimation primitives do not themselves introduce inverse
dependence on $\mu_1$, the end-to-end QKS complexity does not inherit the
divergence associated with the retained small eigenvalue.
\end{theorem}

\paragraph{Proof sketch.}
Theorem~\ref{thm:structural} provides an accurate vector
$z_m\in\mathcal K_m$ and simultaneously bounds the structural quantities
$m$, $T$, and $p_{\rm succ}^{-1}$ independently of the retained small
eigenvalue.  Theorem~\ref{thm:solution_aligned} then shows that, although
$\|B_m^{-1}\|_2=\Theta(\mu_1^{-1})$, the normalized directional
sensitivities satisfy
\begin{equation}
    K_B=O(1),
    \qquad
    K_c=O(\mu_1).
\end{equation}
Theorem~\ref{thm:galerkin_transfer} uses this bounded directional
sensitivity to transfer the subspace-approximation accuracy to the
Galerkin solution actually returned by QKS.

Finally, Theorems~\ref{thm:SHg_Bc} and
\ref{thm:finite_stability} show that finite errors in the measured
projected quantities can be propagated through the orthonormal reduced
system and into the normalized solution without introducing an
inverse-$\mu_1$ precision requirement.  Combining these results controls
every condition-sensitive term in the general QKS complexity bound.
The detailed ingredients are proved in
Secs.~\ref{sec:structural}--\ref{sec:directional}; the complete
composition argument is given in
Appendix~\ref{app:combined_proof}.

Theorem~\ref{thm:combined} removes the retained-small-eigenvalue
divergence from the intrinsic QKS resources.  What remains is the cost of
estimating the finite projected system to the required accuracy.  This
remaining cost depends on the access model and estimator rather than on
the directional conditioning mechanism itself.

\subsection{End-to-end complexity after directional compression}
\label{subsec:e2e_complexity}

We now insert the bounds of
Theorem~\ref{thm:combined} into the general end-to-end QKS complexity
formula of Theorem~\ref{thm:general_complexity}.  Since
$m$, $T$, $p_{\rm succ}^{-1}$, and $K_B$ remain bounded with respect to
$\mu_1$, and since $K_c$ decreases rather than diverges as
$\mu_1\rightarrow0$, the only remaining nonconstant contribution is the
implementation-dependent cost of estimating the projected system.

\begin{corollary}[QKS complexity beyond global conditioning]
\label{cor:e2e}

Under the hypotheses of Theorem~\ref{thm:combined},
\begin{equation}
    \boxed{
    C_{\rm QKS}
    =
    \widetilde O\!\left[
    C_{\rm proj}^{\rm dir}(\epsilon)
    +1
    \right],}
    \label{eq:e2e_reduced_complexity}
\end{equation}
where $C_{\rm proj}^{\rm dir}(\epsilon)$ collects the remaining
implementation- and target-accuracy-dependent cost of estimating the
projected quantities.

The directional analysis implies that
$C_{\rm proj}^{\rm dir}(\epsilon)$ acquires no inverse dependence on
$\mu_1$ through its required estimation precision.  Therefore any remaining
dependence is due to the chosen access model, estimator, system size, and
target accuracy rather than to the retained-small-eigenvalue instability.
\end{corollary}

\paragraph{Proof sketch.}
Theorem~\ref{thm:combined} makes every structural and directional factor
in Theorem~\ref{thm:general_complexity} bounded with respect to $\mu_1$.
Corollary~\ref{cor:no_inverse_precision} further guarantees that the
precision required to estimate $(H_m,\mathbf g_m)$ does not scale with
the inverse retained eigenvalue.  Substitution into the general complexity
expression therefore leaves only the access- and estimator-dependent
projected-system cost summarized by
$C_{\rm proj}^{\rm dir}(\epsilon)$.
The complete proof is given in
Appendix~\ref{app:e2e_corollary_proof}.

This corollary establishes the access-independent part of the main
complexity result.  To understand what complexity remains in an actual
implementation, we next instantiate
$C_{\rm proj}^{\rm dir}(\epsilon)$ under representative access models.

\subsection{Access- and estimator-dependent projected-system costs}
\label{subsec:access_models}

The preceding results deliberately separate the beyond-conditioning
mechanism from the physical procedure used to estimate the projected
system.  We now make this remaining implementation dependence explicit.

Corollary~\ref{cor:no_inverse_precision} allows the actually estimated
projected quantities to satisfy
\begin{equation}
    \|\Delta H_m\|_2,
    \qquad
    \|\Delta\mathbf g_m\|_2
    =
    O(\epsilon),
    \label{eq:access_precision_budget}
\end{equation}
where the hidden constants are independent of $\mu_1$.

The QKS time-translation structure further reduces the number of distinct
projected-operator quantities that must be estimated.  Specifically,
\begin{equation}
    [H_m]_{jk}
    =
    \langle b|
    A e^{-i(k-j)\tau A}
    |b\rangle.
    \label{eq:lag_structure_SH}
\end{equation}
Thus $H_m$ depends only on the time lag $k-j$ and contains only $O(m)$
distinct correlations, while $\mathbf g_m$ has $m$ components.

To guarantee the norm-level precision in
Eq.~\eqref{eq:access_precision_budget}, it is sufficient to estimate
\begin{equation}
    \xi_H
    =
    \Theta(\epsilon/m),
    \qquad
    \xi_g
    =
    \Theta(\epsilon/\sqrt m).
    \label{eq:entrywise_precision_budget}
\end{equation}
The first relation ensures an $O(\epsilon)$ matrix-norm error, whereas the
second follows from the Euclidean norm of the $m$-component projected
right-hand side.

Because $m=O(1)$ with respect to $\mu_1$, none of these required entrywise
precisions becomes more stringent as the retained eigenvalue vanishes.
What changes from one implementation to another is only the cost of
achieving these fixed accuracies.

Representative projected-system costs are summarized below; detailed
derivations are given in
Appendix~\ref{app:access_model_derivations}.

\begin{center}
\small
\begin{tabular}{p{0.22\linewidth}p{0.22\linewidth}p{0.32\linewidth}p{0.15\linewidth}}
\hline
Access model & Projected-system estimator &
Dominant projected-system query cost &
Fixed-$m$ precision \\
\hline

Block encoding
&
Hadamard-test sampling
&
$\widetilde O\!\left[
(\alpha T+1)(1+\alpha^2)m^3/\epsilon^2
\right]$
&
$\epsilon^{-2}$
\\[1mm]

Block encoding
&
Coherent amplitude/expectation estimation
&
$\widetilde O\!\left[
(\alpha T+1)(1+\alpha)m^2/\epsilon
\right]$
&
$\epsilon^{-1}$
\\[1mm]

$s$-sparse oracle
&
Xu short-time finite difference +
coherent multi-expectation estimation
&
$\widetilde O\!\left[
sT\,\operatorname{polylog}N\,
m^{5/2}/\epsilon^2
\right]$
&
$\epsilon^{-2}$
\\[1mm]

Pauli/LCU,
$A=\sum_{\ell=1}^{L}a_\ell P_\ell$
&
Direct projected-matrix estimation +
coherent multi-expectation estimation
&
$\widetilde O\!\left[
(\Lambda_A T+1)
\Lambda_A m^{3/2}\sqrt L/\epsilon
\right]$
&
$\epsilon^{-1}$
\\

\hline
\end{tabular}
\end{center}

Here
\begin{equation}
    \Lambda_A
    :=
    \sum_{\ell=1}^{L}|a_\ell|.
\end{equation}
The table suppresses lower-order projected-$\mathbf g_m$ contributions
as well as polylogarithmic simulation-precision factors.

The table already shows the main separation: the access model can change
the dependence on $\epsilon$, system size, and representation parameters,
but it does not alter the absence of inverse-$\mu_1$ directional
conditioning.  We next derive each representative scaling.

\paragraph{Block-encoding access.}

We first consider block encoding, for which projected quantities can be
estimated either by direct sampling or by coherent expectation estimation.

Assume a unitary $U_b$ preparing $|b\rangle$ and an
$(\alpha,a,\delta)$ block encoding of $A$.
Qubitization-based evolution to time $T$ has query cost
\begin{equation}
    \widetilde O(\alpha T+1)
\end{equation}
with only polylogarithmic dependence on the simulation precision
~\cite{LowChuang2019Qubitization,Gilyen2019}.

Direct sampling of the $O(m)$ distinct lags gives
\begin{equation}
    \boxed{
    C_{\rm proj}^{\rm BE,samp}
    =
    \widetilde O\!\left[
    (\alpha T+1)
    \left(
    \frac{\alpha^2m^3}{\epsilon^2}
    +
    \frac{m^2}{\epsilon^2}
    \right)
    \right].}
    \label{eq:block_encoding_sampling_cost}
\end{equation}
The factor $\alpha^2$ arises because estimating an $H_m$ element to
absolute accuracy $\xi_H$ requires estimating the corresponding normalized
block-encoded matrix element to accuracy $\xi_H/\alpha$.

If coherent amplitude or expectation estimation is available
~\cite{Brassard2002,Huggins2022}, the sampling dependence
$1/\xi^2$ is replaced by $1/\xi$, yielding
\begin{equation}
    \boxed{
    C_{\rm proj}^{\rm BE,coh}
    =
    \widetilde O\!\left[
    (\alpha T+1)
    \left(
    \frac{\alpha m^2}{\epsilon}
    +
    \frac{m^{3/2}}{\epsilon}
    \right)
    \right].}
    \label{eq:block_encoding_coherent_cost}
\end{equation}

Thus block encoding preserves the beyond-conditioning mechanism; the
remaining distinction is whether the projected expectation values are
estimated with sampling or coherent precision scaling.  We next examine
the short-time sparse-access construction, where an additional
finite-difference precision penalty appears.

\paragraph{Sparse access and the short-time estimator.}

Under normalized $s$-sparse Hamiltonian access, the short-time
construction of Ref.~\cite{Xu2024} estimates the projected operator through
\begin{equation}
    [H_m]_{jk}
    =
    \frac{i}{h}
    \left[
    G_{jk}(h)-G_{jk}(0)
    \right]
    +
    O(h),
\end{equation}
where
\begin{equation}
    G_{jk}(h)
    :=
    \langle\phi_j|
    e^{-ihA}
    |\phi_k\rangle.
    \label{eq:xu_access_model_estimator}
\end{equation}

To obtain an $H_m$-element accuracy $\xi_H$, one chooses
\begin{equation}
    h=\Theta(\xi_H),
    \qquad
    \delta_H=\Theta(\xi_H^2),
    \label{eq:xu_access_precision}
\end{equation}
where $\delta_H$ is the accuracy required for the underlying dynamical
correlation.  The extra power of precision arises from the
$1/h$ finite-difference amplification.

Because only $O(m)$ distinct lags are needed, coherent
multiple-expectation estimation~\cite{Huggins2022} yields
\begin{equation}
    \boxed{
    C_{\rm proj}^{\rm sparse,FD}
    =
    \widetilde O\!\left[
    sT\,\operatorname{polylog}N
    \left(
    \frac{m^{5/2}}{\epsilon^2}
    +
    \frac{m}{\epsilon}
    \right)
    \right].}
    \label{eq:sparse_fd_access_cost}
\end{equation}
The first term is the projected-$H_m$ contribution and contains the
finite-difference precision penalty; the remaining term comes from direct
estimation of $\mathbf g_m$.

Hence sparse short-time access changes the target-accuracy scaling, but
still does not reintroduce inverse dependence on the retained soft
eigenvalue.  We next contrast this with direct Pauli/LCU projection, which
avoids the finite-difference step.

\paragraph{Pauli/LCU access and direct projection.}

Suppose
\begin{equation}
    A
    =
    \sum_{\ell=1}^{L}
    a_\ell P_\ell,
    \qquad
    \Lambda_A
    :=
    \sum_{\ell=1}^{L}|a_\ell|.
    \label{eq:pauli_access_decomposition}
\end{equation}
Then the projected operator can be written directly as
\begin{equation}
    [H_m]_{jk}
    =
    \sum_{\ell=1}^{L}
    a_\ell
    \langle b|
    P_\ell e^{-i(k-j)\tau A}
    |b\rangle.
    \label{eq:pauli_direct_projected_H}
\end{equation}

There are $O(mL)$ distinct term-level correlations.  Estimating each to
\begin{equation}
    \delta_P
    =
    \Theta(\xi_H/\Lambda_A)
\end{equation}
and using coherent multiple-expectation estimation gives
\begin{equation}
    \boxed{
    C_{\rm proj}^{\rm Pauli}
    =
    \widetilde O\!\left[
    (\Lambda_A T+1)
    \left(
    \frac{\Lambda_A m^{3/2}\sqrt L}{\epsilon}
    +
    \frac{m}{\epsilon}
    \right)
    \right].}
    \label{eq:pauli_access_cost}
\end{equation}

Unlike the finite-difference construction, direct Pauli projection contains
no $1/h$ amplification and therefore does not require an
$O(\xi_H^2)$ intermediate precision.  The difference between the sparse
and Pauli implementations is therefore an access-dependent precision
effect rather than a difference in the beyond-conditioning mechanism.

\subsection{Access-model robustness of the beyond-conditioning mechanism}
\label{subsec:access_robustness}

The preceding examples demonstrate that different matrix-access models can
produce different dependence on target accuracy and representation
parameters.  The final question is which part of the complexity statement
is universal across these implementations.

\begin{corollary}[Access-model robustness of the beyond-conditioning result]
\label{cor:access_model_robustness}

Under the regime of Theorem~\ref{thm:combined}, suppose that the parameters
of the chosen access model---for example $\alpha$ for block encoding,
$s$ for sparse access, or $\Lambda_A$ and $L$ for Pauli/LCU access---do
not themselves scale inversely with the vanishing solution-relevant
eigenvalue.

Then none of
Eqs.~\eqref{eq:block_encoding_sampling_cost},
\eqref{eq:block_encoding_coherent_cost},
\eqref{eq:sparse_fd_access_cost}, or
\eqref{eq:pauli_access_cost}
acquires an inverse dependence on that eigenvalue.

Therefore the access model can modify the remaining dependence on
$\epsilon$, system size, and matrix representation, but it does not alter
the directional beyond-conditioning mechanism established in
Secs.~\ref{sec:structural}--\ref{sec:directional}.
\end{corollary}

\paragraph{Proof sketch.}
Theorem~\ref{thm:combined} gives
$m,T=O(1)$ with respect to $\mu_1$, while
Corollary~\ref{cor:no_inverse_precision} makes the required projected-system
precision independent of $\mu_1$.  Inspecting the representative costs
above therefore leaves only the corresponding access parameters and target
accuracy as possible sources of additional scaling.  If those access
parameters are themselves independent of $\mu_1$, none of the projected
estimation costs inherits an inverse-$\mu_1$ dependence.
The complete proof is given in
Appendix~\ref{app:access_robustness_proof}.

\begin{remark}[What is access dependent and what is not]
The distinction is central to the end-to-end statement.  The disappearance
of the inverse-small-eigenvalue divergence follows from two intrinsic
properties of QKS:
\begin{equation}
    \text{solution-weighted spectral compression}
    \quad\text{and}\quad
    \text{normalized-direction stability}.
\end{equation}
By contrast, the powers of $1/\epsilon$ and the prefactors involving
$\alpha$, $s$, $L$, or $\Lambda_A$ depend on how the projected quantities
are physically represented and estimated.

Thus the result removes the retained-small-eigenvalue contribution to the
asymptotic hardness without claiming that all access models have identical
precision or representation cost.
\end{remark}

Combining the results of this section, the end-to-end picture is therefore
\begin{equation}
    \boxed{
    \begin{aligned}
    &\text{solution-weighted compression}
    &&\Longrightarrow
    &&m,\;T,\;p_{\rm succ}^{-1}=O(1),
    \\
    &\text{directional stability}
    &&\Longrightarrow
    &&K_B=O(1),\;K_c=O(\mu_1),
    \\
    &\text{finite-error robustness}
    &&\Longrightarrow
    &&\text{no inverse-}\mu_1\text{ precision requirement},
    \\
    &\text{access-model instantiation}
    &&\Longrightarrow
    &&\text{only implementation-dependent costs remain}.
    \end{aligned}}
\end{equation}

Hence, in the retained solution-relevant small-eigenvalue regime considered
here, the QKS end-to-end complexity does not inherit the corresponding
conditioning divergence; the remaining complexity is determined by target
accuracy, matrix representation, and the physical access model.
\section{Relation to Existing Beyond-Conditioning Mechanisms}
\label{sec:separation}

The preceding sections establish the central beyond-conditioning mechanism
of this work: a solution-relevant small eigenvalue may remain explicitly
retained by the QKS reduced problem, so that the reduced inverse norm itself
diverges, while the normalized solution direction nevertheless remains
stable.

This section clarifies how this mechanism relates to existing approaches
that avoid global-condition-number dependence.  There are two conceptually
different favorable regimes.  In the first, globally small eigenvalues are
negligible to the target solution and therefore become effectively excluded
from the reduced problem.  This reproduces the familiar intuition behind
effective-truncation methods.  In the second, which is the central regime of
this paper, the small eigenvalue remains solution-relevant and cannot be
discarded; instead, its divergent inverse amplification becomes
asymptotically parallel to the normalized solution direction.

The comparison proceeds in four steps.  We first identify the regime in
which QKS overlaps with effective-truncation mechanisms.  We then compare
the directional sensitivity governing QKS with a full inverse-amplification
criterion appearing in existing beyond-conditioning analysis.  Next, we
construct an explicit family that separates the two criteria.  Finally, we
embed the same family into arbitrarily large Hilbert spaces to show that the
separation is a property of solution-relevant geometry rather than an
artifact of a two-dimensional example.

\subsection{Overlap with effective truncation}
\label{subsec:truncatable_overlap}

We begin with the regime in which arbitrarily small global eigenvalues are
\emph{not} important to the normalized solution.  The purpose of this
subsection is to show that, when the solution-relevant active spectrum
remains bounded away from zero, the QKS subspace has only limited overlap
with the globally ill-conditioned directions.  In this case, the reduced
problem itself becomes well conditioned.

Let
\begin{equation}
    P_{\mathcal A}
    :=
    \sum_{a\in\mathcal A}
    |u_a\rangle\langle u_a|,
    \qquad
    P_{\mathcal I}
    :=
    I-P_{\mathcal A},
\end{equation}
denote the active and inactive spectral projectors.  Recall also the
orthonormal QKS basis
\begin{equation}
    Q_m
    =
    \Phi_m S_m^{-1/2}.
\end{equation}
We quantify the leakage of the QKS subspace into the inactive spectral
sector by
\begin{equation}
    \eta_{\rm leak}
    :=
    \|P_{\mathcal I}Q_m\|_2.
    \label{eq:inactive_leakage_definition}
\end{equation}

\begin{theorem}[Truncatable ill-conditioning yields a well-conditioned reduced problem]
\label{thm:truncatable_branch}

Assume the structural conditions (A1)--(A2), and suppose that every active
eigenvalue satisfies
\begin{equation}
    \lambda_a
    \ge
    \Delta_{\mathcal A}>0,
    \qquad
    a\in\mathcal A,
    \label{eq:active_lower_gap}
\end{equation}
where $\Delta_{\mathcal A}$ is independent of arbitrarily small
eigenvalues outside the active set.

Then
\begin{equation}
    \boxed{
    \eta_{\rm leak}
    \le
    \frac{
    \sqrt m\,\epsilon_{\rm sub}
    }{
    s_0\Delta_{\mathcal A}
    \sqrt{1-\epsilon_{\rm sub}^2}
    }.}
    \label{eq:inactive_subspace_leakage}
\end{equation}

Consequently, if
\begin{equation}
    \eta_{\rm leak}
    \le
    \eta_0<1
\end{equation}
for a fixed constant $\eta_0$, then
\begin{equation}
    \boxed{
    \lambda_{\min}(B_m)
    \ge
    \Delta_{\mathcal A}
    (1-\eta_0^2)
    =
    \Omega(1).}
    \label{eq:truncatable_reduced_gap}
\end{equation}

Thus arbitrarily small eigenvalues outside the active set may drive the
global condition number to infinity without making the retained QKS
reduced problem ill conditioned.
\end{theorem}

\paragraph{Proof sketch.}
Assumption (A1) bounds the inactive contribution to the normalized inverse
response.  Because the active eigenvalues are bounded below by
$\Delta_{\mathcal A}$, this implies a corresponding bound on the inactive
input-state amplitude.  Time evolution preserves these spectral
amplitudes, while the Gram lower bound
$\lambda_{\min}(S_m)\ge s_0^2$ controls the orthonormalization
$Q_m=\Phi_mS_m^{-1/2}$.  Together these observations give
Eq.~\eqref{eq:inactive_subspace_leakage}.

If the QKS subspace has inactive leakage at most $\eta_0$, every normalized
vector in that subspace retains at least $1-\eta_0^2$ of its weight on the
active spectral sector.  The min--max principle then gives
Eq.~\eqref{eq:truncatable_reduced_gap}.  The complete proof is given in
Appendix~\ref{app:truncatable_proof}.

\begin{corollary}[Overlap with effective-truncation beyond-conditioning]
\label{cor:overlap_truncation}

Under the hypotheses of
Theorem~\ref{thm:truncatable_branch}, the global condition number may
diverge because of eigenvalues outside $\mathcal A$, while
\begin{equation}
    m,\qquad
    T,\qquad
    p_{\rm succ}^{-1},\qquad
    K_B
    =
    O(1)
\end{equation}
with respect to those eigenvalues.

Hence QKS contains an overlap regime in which globally ill-conditioned
directions are negligible to the normalized solution and the retained
reduced problem remains uniformly well conditioned.
The complete proof is given in
Appendix~\ref{app:overlap_corollary_proof}.
\end{corollary}

This regime is conceptually related to effective-truncation approaches to
beyond-conditioning linear solving~\cite{Dalzell2026}: in both cases,
globally small eigenvalues cease to determine the relevant complexity
because their contribution to the target solution is negligible.  The
favorable classes are not identical, however, because QKS additionally
requires compact time-evolution representability and stable Krylov-basis
geometry.

The truncatable regime therefore explains one familiar route by which QKS
can avoid global conditioning.  It is not the central mechanism developed
in this paper.  We next return to the harder regime in which the small
eigenvalue remains retained and solution-relevant, so that the reduced
inverse norm itself still diverges.

\subsection{Directional sensitivity versus full inverse amplification}
\label{subsec:directional_vs_inverse}

Once the small eigenvalue is retained, the reduced problem itself remains
ill conditioned and the effective-truncation explanation no longer applies.
The relevant question is then which part of the inverse amplification can
actually change the normalized target state.

For the normalized target state
\begin{equation}
    |x\rangle
    =
    \frac{A^{-1}|b\rangle}
    {\|A^{-1}|b\rangle\|_2},
\end{equation}
the filtering-based analysis of
Ref.~\cite{Dalzell2026} contains the inverse-amplification quantity
\begin{equation}
    \chi_{\rm filt}(A,b)
    :=
    \|A^{-1\dagger}|x\rangle\|_2.
    \label{eq:su_filter_measure}
\end{equation}
This quantity retains the full action of the inverse on the normalized
solution.

By contrast, the directional QKS sensitivity derived in
Sec.~\ref{sec:directional} is
\begin{equation}
    K_B
    =
    \|P_{\mathbf y}^{\perp}B_m^{-1}\|_2,
\end{equation}
where
\begin{equation}
    P_{\mathbf y}^{\perp}
    =
    I-
    \frac{\mathbf y\mathbf y^\dagger}
    {\|\mathbf y\|_2^2}.
\end{equation}
The projector explicitly removes inverse amplification parallel to the
solution direction.

The distinction is therefore geometric:
\begin{equation}
    \boxed{
    \text{full inverse amplification}
    \;\neq\;
    \text{normalized directional sensitivity}.}
    \label{eq:geometric_distinction}
\end{equation}
Amplification parallel to the solution may become arbitrarily large without
changing the normalized state, whereas transverse amplification directly
changes its direction.

This difference suggests that the two criteria need not have the same
asymptotic behavior.  The next subsection shows that this distinction is
genuine by constructing a family on which one quantity diverges while the
other remains bounded.

\subsection{Explicit separation of the two criteria}
\label{subsec:explicit_separation}

We now construct the simplest explicit family that separates full inverse
amplification from normalized directional sensitivity.

\begin{theorem}[Separation of inverse-amplification and directional criteria]
\label{thm:separation_su}

Consider
\begin{equation}
    A_\varepsilon
    =
    \begin{pmatrix}
        \varepsilon & 0\\
        0 & 1
    \end{pmatrix},
    \qquad
    |b\rangle
    =
    \frac1{\sqrt2}
    \begin{pmatrix}
        1\\
        1
    \end{pmatrix},
    \qquad
    0<\varepsilon<1.
    \label{eq:separation_family}
\end{equation}

Then
\begin{equation}
    \kappa(A_\varepsilon)
    =
    \varepsilon^{-1}
    \rightarrow\infty,
\end{equation}
while the normalized solution satisfies
\begin{equation}
    |x_\varepsilon\rangle
    =
    \frac{(1,\varepsilon)^T}
    {\sqrt{1+\varepsilon^2}}
    \longrightarrow
    \begin{pmatrix}
        1\\
        0
    \end{pmatrix}.
\end{equation}

The full inverse-amplification quantity obeys
\begin{equation}
    \boxed{
    \chi_{\rm filt}(A_\varepsilon,b)
    =
    \frac{
    \sqrt{\varepsilon^{-2}+\varepsilon^2}
    }{
    \sqrt{1+\varepsilon^2}
    }
    =
    \Theta(\varepsilon^{-1}).}
    \label{eq:su_measure_diverges}
\end{equation}

By contrast, choose any fixed
\begin{equation}
    \tau\notin2\pi\mathbb Z.
\end{equation}
For sufficiently small $\varepsilon$, the two-state QKS basis remains
uniformly nondegenerate and spans the full two-dimensional relevant
subspace.  Consequently,
\begin{equation}
    \boxed{
    m=2,\qquad
    T=O(1),\qquad
    p_{\rm succ}^{-1}=O(1),}
    \label{eq:separation_structural}
\end{equation}
while the corresponding orthonormal reduced system satisfies
\begin{equation}
    \mu_1=\varepsilon,
    \qquad
    \mu_2=1,
    \qquad
    |c_1|,|c_2|=\Theta(1).
\end{equation}
Therefore
\begin{equation}
    \boxed{
    K_B=O(1),
    \qquad
    K_c=O(\varepsilon).}
    \label{eq:qks_directional_separation}
\end{equation}

Hence the full inverse-amplification measure diverges as
$\Theta(\varepsilon^{-1})$, while the structural and directional QKS
quantities identified in the preceding analysis remain bounded with
respect to that divergence.
\end{theorem}

\paragraph{Proof sketch.}
The inverse
\begin{equation}
    A_\varepsilon^{-1}
    =
    \operatorname{diag}(\varepsilon^{-1},1)
\end{equation}
amplifies the soft eigendirection by $\varepsilon^{-1}$.  Applying it to
the normalized solution therefore gives
$\chi_{\rm filt}=\Theta(\varepsilon^{-1})$.

For QKS, the two columns
\begin{equation}
    |b\rangle,
    \qquad
    e^{-i\tau A_\varepsilon}|b\rangle
\end{equation}
remain uniformly linearly independent as $\varepsilon\to0$ whenever
$\tau\notin2\pi\mathbb Z$.  The corresponding Gram matrix therefore has a
strictly positive lower spectral bound independent of $\varepsilon$.
Since the resulting QKS subspace is the full relevant two-dimensional
space, the orthonormal reduced problem retains the eigenvalues
$\varepsilon$ and $1$, and the reduced right-hand side retains constant
overlap with the soft eigendirection.  Theorem~\ref{thm:solution_aligned}
then gives Eq.~\eqref{eq:qks_directional_separation}.

The complete proof is given in
Appendix~\ref{app:separation_su_proof}.

The theorem establishes that bounded normalized-direction sensitivity is
not equivalent to bounded full inverse amplification.  However, because
the example is two-dimensional, one could still ask whether the separation
is caused only by the QKS subspace coinciding with the full Hilbert space.
The next subsection rules out this possibility.

\subsection{High-dimensional embedding of the separation}
\label{subsec:high_dimensional_separation}

The purpose of this subsection is to show that the preceding separation
persists even when the QKS subspace is much smaller than the ambient
Hilbert space.

\begin{proposition}[High-dimensional embedding of the separation]
\label{prop:high_dimensional_separation}

For any $N\ge2$, define
\begin{equation}
    A_{\varepsilon,N}
    :=
    A_\varepsilon\oplus D_{N-2},
    \qquad
    |b_N\rangle
    :=
    |b\rangle\oplus0,
    \label{eq:high_dimensional_embedding}
\end{equation}
where $A_\varepsilon$ and $|b\rangle$ are given by
Eq.~\eqref{eq:separation_family}, and $D_{N-2}$ is any positive-definite
matrix whose spectrum lies in
\begin{equation}
    [\Delta_{\rm sp},1],
    \qquad
    \Delta_{\rm sp}>0,
\end{equation}
independently of $\varepsilon$ and $N$.

Then
\begin{equation}
    \kappa(A_{\varepsilon,N})
    =
    \Theta(\varepsilon^{-1}),
\end{equation}
and
\begin{equation}
    \chi_{\rm filt}(A_{\varepsilon,N},b_N)
    =
    \Theta(\varepsilon^{-1}),
\end{equation}
whereas the QKS dynamics remain confined to the two-dimensional invariant
subspace supporting $A_\varepsilon$.

Consequently,
\begin{equation}
    \boxed{
    m=2\ll N,\qquad
    T=O(1),\qquad
    p_{\rm succ}^{-1}=O(1),\qquad
    K_B=O(1),\qquad
    K_c=O(\varepsilon).}
    \label{eq:high_dimensional_qks_scaling}
\end{equation}
\end{proposition}

\paragraph{Proof sketch.}
The initial state has support only on the first two-dimensional block, and
the block-diagonal Hamiltonian preserves this invariant subspace under all
time evolutions.  Therefore the entire QKS construction is identical to
that of Theorem~\ref{thm:separation_su}, independently of the ambient
dimension $N$.

The additional block $D_{N-2}$ neither changes the relevant inverse response
nor removes the physical eigenvalue $\varepsilon$.  Hence the global
condition number and the full inverse-amplification criterion retain their
$\Theta(\varepsilon^{-1})$ scaling, while the QKS structural and
directional quantities are unchanged.

The complete proof is given in
Appendix~\ref{app:highdim_proof}.

This proposition removes the dimensional objection: the separation persists
even when $m\ll N$.  We can therefore interpret it as a property of the
solution-relevant spectral geometry rather than a consequence of solving a
two-dimensional full-space problem.

\subsection{Interpretation of the separation}
\label{subsec:separation_interpretation}

The results above place the present QKS mechanism relative to existing
beyond-conditioning approaches in two complementary ways.

First, in the truncatable regime,
\begin{equation}
    \boxed{
    \text{globally small eigenvalues are solution-irrelevant},
    }
\end{equation}
and QKS recovers the familiar effective-conditioning picture: the reduced
problem itself remains uniformly well conditioned.

Second, in the directional regime central to this paper,
\begin{equation}
    \boxed{
    \text{the small eigenvalue is retained and solution-relevant},
    }
\end{equation}
so the reduced inverse norm diverges.  Nevertheless, the divergent inverse
amplification becomes increasingly parallel to the normalized solution
direction, while the transverse sensitivity remains bounded.

Theorem~\ref{thm:separation_su} and
Proposition~\ref{prop:high_dimensional_separation} do not imply that QKS
dominates all beyond-condition-number algorithms or every favorable regime
of Ref.~\cite{Dalzell2026}.  Rather, they establish the more specific
separation
\begin{equation}
    \boxed{
    \text{full inverse amplification can diverge}
    \quad\text{while}\quad
    \text{normalized directional sensitivity remains bounded}.}
    \label{eq:final_directional_separation}
\end{equation}

The beyond-conditioning mechanism developed here is therefore not merely an
implicit truncation of globally small eigenvalues.  It identifies a
distinct favorable regime in which a retained small eigenvalue may dominate
the solution magnitude and make the reduced inverse norm diverge, without
destabilizing the normalized solution direction.
\section{Discussion}
\label{sec:discussion}

The main result shows that a solution-relevant vanishing eigenvalue need not
force increasing normalized-state difficulty in time-evolution QKS.  The
structural requirement is that the relevant inverse response admit a compact
and stable time-evolution representation; after this compression, the key
question is geometric rather than purely spectral.  A divergent inverse
component that becomes parallel to the normalized solution changes the norm
of the response without producing a divergent change in its direction.

\subsection{Physical interpretation and a simple mechanical instability}
\label{sec:regime}

Quantum linear-system solvers also arise as primitives in differential
equations and Green's-function calculations
~\cite{BerryCosta2024,Krovi2023,Tong2021}.  A complementary physical family
relevant to the present directional mechanism arises in linearized structural
response near a simple mechanical instability, where loss of stability is
associated with a vanishing stiffness mode~\cite{Dalklint2021}.  In the
single-mode case, one critical stiffness mode approaches zero
while the remaining spectrum stays nonsingular.  If the applied load has
nonzero overlap with that critical mode, the response becomes increasingly
dominated by it.  This realizes the geometry behind (B1)--(B2): the
unnormalized response grows strongly along one soft direction, while its
normalized direction becomes increasingly stable.

The two-dimensional family in Eq.~\eqref{eq:separation_family} is the
minimal illustration of this single-soft-mode geometry.  The overlap regime
of Section~\ref{sec:separation} has a different interpretation: very soft
modes may exist in the full operator but the chosen forcing excites them too
weakly to contribute appreciably to the normalized response.  These two
situations explain why global conditioning can be pessimistic for physically
different reasons.

\subsection{Boundaries and scope}

The directional protection can fail through several distinct mechanisms.
First, if active phases or significant inactive spectral mass become
increasingly difficult to resolve, the required Krylov dimension or
evolution horizon may grow; this is the spectral-resolution boundary
captured by (A2).  Second, if the Krylov Gram geometry becomes unstable, the
orthonormalization and reconstruction costs can become unfavorable.  Third,
if several independent reduced directions become singular, normalization
removes only the component parallel to the solution, while transverse
singular directions may make $K_B$ diverge.  The present theorem therefore
does not claim condition-number-independent behavior for arbitrary
multi-mode singular limits.

The relation to existing beyond-condition-number methods should also be read
as partial overlap rather than subsumption.  Section~\ref{sec:separation}
shows that QKS recovers an effective-truncation-type regime when globally
small directions are solution-negligible, but QKS additionally requires
compact time-evolution representability and stable basis geometry.  The
explicit separation demonstrates a distinct directional mechanism, not a
blanket complexity improvement over all existing algorithms.

\subsection{Implementation dependence and outlook}

The finite-error and coordinate-transfer results close the measurement-error
chain needed for the end-to-end theorem, while the access-model comparison
separates the universal conditioning mechanism from
implementation-dependent costs.  Block-encoding, sparse-oracle, and
Pauli/LCU realizations lead to different precision and representation
overheads, but none restores inverse dependence on the retained vanishing
eigenvalue when its own access parameters remain bounded.

The result should not be interpreted as a general removal of condition
number from quantum linear solving.  Rather, it shows that the global
condition number alone is an incomplete descriptor of normalized QKS
difficulty.  The relevant quantities are solution-weighted spectral
representability and the geometry of inverse amplification relative to the
normalized solution direction.

Natural extensions include structured multi-mode directional stability,
weaker spectral-resolvability conditions, and application-specific regimes
in which the favorable quantities can be certified directly.  Improved
matrix access and projected-system estimators can reduce the remaining
representation and precision costs without changing the directional
mechanism identified here.

\appendix

\section{General End-to-End QKS Complexity}
\label{app:general_complexity_proof}

This appendix derives the complexity decomposition in
Theorem~\ref{thm:general_complexity}.  The derivation separates four costs:
projected-operator estimation, projected-right-hand-side estimation, the
classical reduced solve, and quantum state reconstruction.

Throughout this derivation, $S_m$ is used only to define the equivalent
orthonormal coordinates for theoretical analysis.  It does not contribute
a separate estimation cost.  When
\begin{equation}
    \lambda_{\min}(S_m)\ge \sigma_S>0,
\end{equation}
Theorem~\ref{thm:SHg_Bc} shows that perturbations of the actually estimated
quantities $(H_m,\mathbf g_m)$ induce perturbations of
$(B_m,\mathbf c_m)$ with prefactors depending only on $\sigma_S$ and not on
the vanishing reduced eigenvalue.  These fixed Gram-stability factors are
suppressed in the schematic complexity expression of
Theorem~\ref{thm:general_complexity}.

\subsection{Normalized-state sensitivity of the reduced solve}

Consider the exact orthonormal reduced system
\begin{equation}
    B_m\mathbf y=\mathbf c_m,
    \qquad
    \mathbf y=B_m^{-1}\mathbf c_m,
\end{equation}
and let
\begin{equation}
    P_{\mathbf y}^{\perp}
    =
    I-\frac{\mathbf y\mathbf y^\dagger}{\|\mathbf y\|_2^2}.
\end{equation}
Suppose that estimation of the projected system produces perturbations
\begin{equation}
    \widehat B_m=B_m+E,
    \qquad
    \widehat{\mathbf c}_m=\mathbf c_m+\mathbf f.
\end{equation}
Writing
\begin{equation}
    \widehat{\mathbf y}
    =
    \widehat B_m^{-1}\widehat{\mathbf c}_m
    =
    \mathbf y+\delta\mathbf y,
\end{equation}
the perturbed linear system gives
\begin{equation}
    (B_m+E)(\mathbf y+\delta\mathbf y)
    =
    \mathbf c_m+\mathbf f.
\end{equation}
Using $B_m\mathbf y=\mathbf c_m$ and retaining the first-order terms,
\begin{equation}
    B_m\delta\mathbf y
    =
    \mathbf f-E\mathbf y,
\end{equation}
and therefore
\begin{equation}
    \delta\mathbf y
    =
    B_m^{-1}\mathbf f
    -
    B_m^{-1}E\mathbf y
    +
    O\!\left(
    (\|E\|_2+\|\mathbf f\|_2)^2
    \right).
    \label{eq:first_order_reduced_perturbation}
\end{equation}

For normalized-state preparation, only the component of
$\delta\mathbf y$ transverse to $\mathbf y$ changes the physical state to
first order.  Indeed,
\begin{equation}
    \sin\Theta(\mathbf y,\widehat{\mathbf y})
    =
    \frac{
    \|P_{\mathbf y}^{\perp}\delta\mathbf y\|_2
    }{
    \|\mathbf y+\delta\mathbf y\|_2
    }.
\end{equation}
Hence, in the perturbative regime,
\begin{equation}
    \sin\Theta(\mathbf y,\widehat{\mathbf y})
    =
    \frac{
    \|P_{\mathbf y}^{\perp}\delta\mathbf y\|_2
    }{
    \|\mathbf y\|_2
    }
    +
    O\!\left(
    (\|E\|_2+\|\mathbf f\|_2)^2
    \right).
\end{equation}
Substituting Eq.~\eqref{eq:first_order_reduced_perturbation} gives
\begin{align}
    \sin\Theta(\mathbf y,\widehat{\mathbf y})
    \le{}&
    \frac{
    \|P_{\mathbf y}^{\perp}B_m^{-1}E\mathbf y\|_2
    }{
    \|\mathbf y\|_2
    }
    +
    \frac{
    \|P_{\mathbf y}^{\perp}B_m^{-1}\mathbf f\|_2
    }{
    \|\mathbf y\|_2
    }
    \nonumber\\
    &+
    O\!\left(
    (\|E\|_2+\|\mathbf f\|_2)^2
    \right).
\end{align}
Using
\begin{equation}
    K_B
    :=
    \|P_{\mathbf y}^{\perp}B_m^{-1}\|_2,
    \qquad
    K_c
    :=
    \frac{K_B}{\|\mathbf y\|_2},
\end{equation}
we obtain
\begin{equation}
    \boxed{
    \sin\Theta(\mathbf y,\widehat{\mathbf y})
    \le
    K_B\|E\|_2
    +
    K_c\|\mathbf f\|_2
    +
    O\!\left(
    (\|E\|_2+\|\mathbf f\|_2)^2
    \right).}
    \label{eq:local_state_sensitivity}
\end{equation}

Equation~\eqref{eq:local_state_sensitivity} explains the two sensitivity
parameters appearing in Theorem~\ref{thm:general_complexity}.  To allocate
a constant fraction of the target normalized-state error $\epsilon$ to each
source, it is sufficient, up to constant factors, to require
\begin{equation}
    \|E\|_2
    =
    O\!\left(\frac{\epsilon}{K_B}\right),
    \qquad
    \|\mathbf f\|_2
    =
    O\!\left(\frac{\epsilon}{K_c}\right).
    \label{eq:operator_vector_precision}
\end{equation}

The argument above is a local sensitivity calculation.  For the
retained-small-eigenvalue regime studied in the main text,
Theorem~\ref{thm:finite_stability} strengthens it to a finite-radius
statement and shows that the neighborhood over which the estimate is valid
does not shrink proportionally to the vanishing eigenvalue.

\subsection{From norm precision to entrywise precision}

The projected operator contains $m^2$ entries.  If each estimated matrix
entry has additive error at most $\eta_H$, then
\begin{equation}
    \|\Delta H_m\|_2
    \le
    \|\Delta H_m\|_F
    \le
    m\eta_H.
    \label{eq:matrix_entry_to_norm}
\end{equation}
Under stable orthonormalization, the corresponding perturbation of $B_m$
differs only by factors independent of the vanishing reduced eigenvalue.
Therefore Eq.~\eqref{eq:operator_vector_precision} is achieved by choosing
\begin{equation}
    \boxed{
    \eta_H
    =
    O\!\left(
    \frac{\epsilon}{mK_B}
    \right).}
    \label{eq:H_entry_precision_general}
\end{equation}

Similarly, if every component of the $m$-dimensional projected
right-hand side is estimated to additive precision $\eta_g$, then
\begin{equation}
    \|\Delta\mathbf g_m\|_2
    \le
    \sqrt m\,\eta_g.
    \label{eq:vector_entry_to_norm}
\end{equation}
Stable orthonormalization then gives
\begin{equation}
    \boxed{
    \eta_g
    =
    O\!\left(
    \frac{\epsilon}{\sqrt m\,K_c}
    \right).}
    \label{eq:g_entry_precision_general}
\end{equation}

Let $\mathcal Q_H(m,\eta)$ denote the number of quantum calls required by
the chosen estimator to obtain the projected-operator data to additive
precision $\eta$, and let $\mathcal Q_g(m,\eta)$ denote the analogous
quantity for the projected right-hand side.  If $C_H(T)$ and $C_g(T)$ are
the corresponding per-call costs, the projected-system construction costs
are
\begin{equation}
    C_H(T)\,
    \mathcal Q_H\!\left(
        m,\frac{\epsilon}{mK_B}
    \right)
\end{equation}
and
\begin{equation}
    C_g(T)\,
    \mathcal Q_g\!\left(
        m,\frac{\epsilon}{\sqrt m\,K_c}
    \right),
\end{equation}
respectively.

For the time-evolution Krylov basis
$|\phi_j\rangle=e^{-ij\tau A}|b\rangle$, the largest evolution time
appearing in any projected quantity is
\begin{equation}
    T=(m-1)\tau,
\end{equation}
which explains the dependence of the per-call costs on $T$.

\subsection{Classical reduced-system cost}

Once the projected data have been estimated, the remaining reduced problem
$H_m\boldsymbol\gamma=\mathbf g_m$ has dimension $m$.  Solving this dense
linear system by standard classical linear algebra requires
\begin{equation}
    O(m^3)
\end{equation}
arithmetic operations.  This accounts for the third term in
Eq.~\eqref{eq:general_qks_complexity}.  Faster classical matrix
multiplication could modify this exponent, but $O(m^3)$ is the standard
dense-solve cost and is sufficient for the present complexity
decomposition.

\subsection{LCU reconstruction and success probability}

After solving the projected system, QKS must prepare
\begin{equation}
    x_m
    =
    \Phi_m\boldsymbol\gamma
    =
    \sum_{j=0}^{m-1}\gamma_j|\phi_j\rangle.
\end{equation}
Write
\begin{equation}
    \gamma_j
    =
    |\gamma_j|e^{i\vartheta_j},
    \qquad
    \|\boldsymbol\gamma\|_1
    =
    \sum_{j=0}^{m-1}|\gamma_j|.
\end{equation}
The standard LCU construction prepares the ancilla state
\begin{equation}
    |\alpha\rangle
    =
    \frac{1}{\sqrt{\|\boldsymbol\gamma\|_1}}
    \sum_{j=0}^{m-1}
    \sqrt{|\gamma_j|}\,|j\rangle,
\end{equation}
applies the corresponding controlled time evolutions together with the
phases $e^{i\vartheta_j}$, and then unprepares the ancilla.  Conditioned on
the successful ancilla outcome, the system register contains the
unnormalized state
\begin{equation}
    \frac{1}{\|\boldsymbol\gamma\|_1}
    \sum_{j=0}^{m-1}
    \gamma_j|\phi_j\rangle
    =
    \frac{x_m}{\|\boldsymbol\gamma\|_1}.
\end{equation}
The success probability is therefore
\begin{align}
    p_{\rm succ}
    &=
    \frac{\|x_m\|_2^2}
    {\|\boldsymbol\gamma\|_1^2}
    \nonumber\\
    &=
    \frac{
    \boldsymbol\gamma^\dagger
    \Phi_m^\dagger\Phi_m
    \boldsymbol\gamma
    }{
    \|\boldsymbol\gamma\|_1^2
    }
    \nonumber\\
    &=
    \boxed{
    \frac{
    \boldsymbol\gamma^\dagger
    S_m
    \boldsymbol\gamma
    }{
    \|\boldsymbol\gamma\|_1^2
    }}.
    \label{eq:psucc_derivation}
\end{align}

Preparing and unpreparing the $m$-component coefficient state costs
$O(m)$ using a generic implementation, while the controlled
time-evolution part costs $C_{\rm K}(T)$.  Thus one reconstruction attempt
costs
\begin{equation}
    O\!\left(m+C_{\rm K}(T)\right).
\end{equation}
Using direct postselection and independent repetition, the expected number
of attempts before success is $1/p_{\rm succ}$.  Hence the expected
reconstruction cost is
\begin{equation}
    O\!\left(
    \frac{m+C_{\rm K}(T)}
    {p_{\rm succ}}
    \right).
    \label{eq:reconstruction_cost_derivation}
\end{equation}
Amplitude amplification could replace the repetition overhead by a
different dependence on $p_{\rm succ}$, but is unnecessary in the
favorable regime considered in the main text, where
$p_{\rm succ}=\Omega(1)$.

\subsection{Total complexity}

Adding the projected-system estimation costs, the classical reduced solve,
and the expected reconstruction cost gives
\begin{equation}
\boxed{
 C_{\rm QKS}
 =
 \widetilde O\!\left[
 C_H(T)\mathcal Q_H\!\left(
 m,\frac{\epsilon}{mK_B}
 \right)
 +
 C_g(T)\mathcal Q_g\!\left(
 m,\frac{\epsilon}{\sqrt m\,K_c}
 \right)
 +
 m^3
 +
 \frac{m+C_{\rm K}(T)}{p_{\rm succ}}
 \right].}
 \label{eq:general_qks_complexity_appendix}
\end{equation}
The $\widetilde O$ notation suppresses constant error-budget allocations,
confidence-amplification factors, and polylogarithmic precision overheads.

Consequently, once implementation-specific access costs are kept explicit,
the quantities through which the vanishing eigenvalue can affect the
end-to-end QKS complexity are the structural resources
\begin{equation}
    m,\qquad T,\qquad p_{\rm succ}^{-1},
\end{equation}
together with the normalized-state sensitivity factors
\begin{equation}
    K_B,\qquad K_c,
\end{equation}
provided the Krylov Gram geometry remains stable.  The structural theorem
and the directional-stability results in the main text establish precisely
this stability in the favorable regime.

\section{Proofs for Spectral Compression and Compact Representation}
\label{app:structural_proofs}

\subsection{Proof of the weighted exponential representation}
\label{app:weighted_exp_proof}

\begin{proof}[Proof of Proposition~\ref{prop:weighted_exp}]
Using the spectral decomposition
\begin{equation}
 A
 =
 \sum_{\ell=1}^{N}
 \lambda_\ell
 |u_\ell\rangle\langle u_\ell|
\end{equation}
and
\begin{equation}
 |b\rangle
 =
 \sum_{\ell=1}^{N}
 b_\ell |u_\ell\rangle,
\end{equation}
the exact unnormalized solution is
\begin{equation}
 x
 =
 A^{-1}|b\rangle
 =
 \sum_{\ell=1}^{N}
 \frac{b_\ell}{\lambda_\ell}
 |u_\ell\rangle.
 \label{eq:appendix_exact_solution}
\end{equation}

The QKS approximation associated with
$\boldsymbol\gamma=(\gamma_0,\ldots,\gamma_{m-1})^T$ is
\begin{equation}
 \Phi_m\boldsymbol\gamma
 =
 \sum_{j=0}^{m-1}
 \gamma_j e^{-ij\tau A}|b\rangle.
\end{equation}
Since each $|u_\ell\rangle$ is an eigenvector of $A$,
\begin{equation}
 e^{-ij\tau A}|u_\ell\rangle
 =
 e^{-ij\tau\lambda_\ell}|u_\ell\rangle.
\end{equation}
Therefore
\begin{align}
 \Phi_m\boldsymbol\gamma
 &=
 \sum_{j=0}^{m-1}
 \gamma_j
 \sum_{\ell=1}^{N}
 b_\ell e^{-ij\tau\lambda_\ell}|u_\ell\rangle
 \nonumber\\
 &=
 \sum_{\ell=1}^{N}
 b_\ell
 \left(
 \sum_{j=0}^{m-1}
 \gamma_j e^{-ij\tau\lambda_\ell}
 \right)
 |u_\ell\rangle.
 \label{eq:appendix_qks_spectral_expansion}
\end{align}

Subtracting Eq.~\eqref{eq:appendix_qks_spectral_expansion} from
Eq.~\eqref{eq:appendix_exact_solution} gives
\begin{equation}
 x-\Phi_m\boldsymbol\gamma
 =
 \sum_{\ell=1}^{N}
 b_\ell
 \left[
 \frac1{\lambda_\ell}
 -
 \sum_{j=0}^{m-1}
 \gamma_j e^{-ij\tau\lambda_\ell}
 \right]
 |u_\ell\rangle.
\end{equation}

Because the eigenvectors $\{|u_\ell\rangle\}_{\ell=1}^{N}$ form an
orthonormal basis, the squared Euclidean norm is the sum of the squared
magnitudes of these coefficients.  Hence
\begin{equation}
 \|x-\Phi_m\boldsymbol\gamma\|_2^2
 =
 \sum_{\ell=1}^{N}
 |b_\ell|^2
 \left|
 \frac1{\lambda_\ell}
 -
 \sum_{j=0}^{m-1}
 \gamma_j e^{-ij\tau\lambda_\ell}
 \right|^2,
\end{equation}
which proves Eq.~\eqref{eq:weighted_exp_identity}.
\end{proof}

\subsection{Proof of the localized time-evolution filter bound}
\label{app:localized_filter_proof}

\begin{proof}[Proof of Lemma~\ref{lem:localized_filter}]
Fix an active mode $a\in\mathcal A$.  By definition,
\begin{equation}
 F_{a,m}(\lambda)
 =
 \frac1m
 \sum_{j=0}^{m-1}
 e^{-ij(\tau\lambda-\theta_a)}.
\end{equation}

We first evaluate the filter at the target eigenvalue $\lambda_a$.
Since
\begin{equation}
 \theta_a
 =
 \tau\lambda_a
 \pmod{2\pi},
\end{equation}
we have
\begin{equation}
 \tau\lambda_a-\theta_a
 =
 0
 \pmod{2\pi}.
\end{equation}
Consequently,
\begin{equation}
 e^{-ij(\tau\lambda_a-\theta_a)}
 =
 1
\end{equation}
for every $j=0,\ldots,m-1$, and therefore
\begin{equation}
 F_{a,m}(\lambda_a)
 =
 \frac1m
 \sum_{j=0}^{m-1}1
 =
 1.
\end{equation}

We next consider a spectral value $\lambda$ whose phase is separated from
$\theta_a$.  Let $\phi\in[-\pi,\pi]$ be the principal representative of
the phase difference
\begin{equation}
 \phi
 =
 \tau\lambda-\theta_a
 \pmod{2\pi}.
\end{equation}
By the definition of the circular distance,
\begin{equation}
 |\phi|
 =
 d_{\mathbb T}(\tau\lambda,\theta_a).
 \label{eq:phase_distance_phi}
\end{equation}
Under the hypothesis
\begin{equation}
 d_{\mathbb T}(\tau\lambda,\theta_a)
 \ge
 \delta>0,
\end{equation}
we therefore have
\begin{equation}
 \delta
 \le
 |\phi|
 \le
 \pi.
 \label{eq:phi_interval}
\end{equation}

Since $\phi\neq0\pmod{2\pi}$, the finite geometric-series identity gives
\begin{equation}
 \sum_{j=0}^{m-1}
 e^{-ij\phi}
 =
 \frac{1-e^{-im\phi}}
 {1-e^{-i\phi}}.
\end{equation}
Hence
\begin{equation}
 F_{a,m}(\lambda)
 =
 \frac{
 1-e^{-im\phi}
 }{
 m(1-e^{-i\phi})
 }.
 \label{eq:filter_geometric_series}
\end{equation}

Taking absolute values,
\begin{equation}
 |F_{a,m}(\lambda)|
 =
 \frac{
 |1-e^{-im\phi}|
 }{
 m|1-e^{-i\phi}|
 }.
\end{equation}
The numerator is bounded by
\begin{equation}
 |1-e^{-im\phi}|
 \le
 2.
\end{equation}
For the denominator, using
\begin{equation}
 |1-e^{-it}|
 =
 2|\sin(t/2)|,
\end{equation}
we obtain
\begin{equation}
 |1-e^{-i\phi}|
 =
 2|\sin(\phi/2)|.
\end{equation}
Therefore
\begin{equation}
 |F_{a,m}(\lambda)|
 \le
 \frac{
 1
 }{
 m|\sin(\phi/2)|
 }.
 \label{eq:filter_pre_gap_bound}
\end{equation}

From Eq.~\eqref{eq:phi_interval},
\begin{equation}
 \frac{\delta}{2}
 \le
 \frac{|\phi|}{2}
 \le
 \frac{\pi}{2}.
\end{equation}
Since $\sin t$ is nondecreasing on $[0,\pi/2]$,
\begin{equation}
 |\sin(\phi/2)|
 =
 \sin(|\phi|/2)
 \ge
 \sin(\delta/2).
\end{equation}
Substituting this inequality into
Eq.~\eqref{eq:filter_pre_gap_bound} gives
\begin{equation}
 |F_{a,m}(\lambda)|
 \le
 \frac1{
 m\sin(\delta/2)
 },
\end{equation}
which proves Eq.~\eqref{eq:filter_bound}.
\end{proof}

\subsection{Proof of the structural QKS representation theorem}
\label{app:structural_proof}

\begin{proof}[Proof of Theorem~\ref{thm:structural}]

The main theorem is stated with abstract universal constants.  The explicit construction below proves one convenient sufficient choice; in particular, the numerical factor $2$ appearing in the interpolation and inactive-mode bounds may be absorbed into $C_{\rm U}$, $C_{\rm res}$, and $C_{\rm app}$.

The proof proceeds in six steps.  We first construct localized filters for
the active modes, then combine them into an interpolant that reproduces the
inverse response exactly on the active spectrum.  We show that the resulting
state belongs to the QKS subspace, bound its leakage onto the inactive
spectrum, and finally establish Gram-matrix and reconstruction stability.

\paragraph{Step 1: Localized filters for the active modes.}

For each active mode $a\in\mathcal A$, recall the localized filter
\begin{equation}
 F_{a,m}(\lambda)
 =
 \frac1m
 \sum_{j=0}^{m-1}
 e^{-ij(\tau\lambda-\theta_a)},
 \qquad
 \theta_a=\tau\lambda_a\pmod{2\pi}.
 \label{eq:structural_filter_recall}
\end{equation}
Lemma~\ref{lem:localized_filter} gives
\begin{equation}
 F_{a,m}(\lambda_a)=1,
 \label{eq:structural_filter_target}
\end{equation}
and, whenever
\begin{equation}
 d_{\mathbb T}(\tau\lambda,\theta_a)\ge\delta,
\end{equation}
\begin{equation}
 |F_{a,m}(\lambda)|
 \le
 \frac1{m\sin(\delta/2)}.
 \label{eq:structural_filter_decay}
\end{equation}

Thus each filter has unit response on its target active mode and suppresses
spectral modes whose time-evolution phases are separated from that target.

\paragraph{Step 2: Stable interpolation of the inverse function on the active spectrum.}

We seek a function $q_m(\lambda)$ that reproduces the inverse function
$1/\lambda$ exactly at every active eigenvalue.  To construct it, define the
$r\times r$ active interpolation matrix
\begin{equation}
 G_{ba}
 :=
 F_{a,m}(\lambda_b),
 \qquad
 a,b\in\mathcal A.
 \label{eq:active_interpolation_matrix}
\end{equation}
Here the column indexed by $a$ records the response of the $a$th localized
filter on all active eigenvalues, while the row indexed by $b$ records the
responses of all filters at the active eigenvalue $\lambda_b$.

By Eq.~\eqref{eq:structural_filter_target},
\begin{equation}
 G_{aa}=1.
\end{equation}
For two distinct active modes $a\neq b$, assumption (A2) gives
\begin{equation}
 d_{\mathbb T}(\theta_a,\theta_b)
 \ge
 \delta_{\mathcal A},
\end{equation}
and therefore Eq.~\eqref{eq:structural_filter_decay} implies
\begin{equation}
 |G_{ba}|
 \le
 \frac1{m\sin(\delta_{\mathcal A}/2)},
 \qquad
 a\neq b.
 \label{eq:G_offdiag_bound}
\end{equation}

Since each column of $G-I$ contains at most $r-1$ nonzero entries,
\begin{equation}
 \|G-I\|_1
 \le
 \frac{r-1}
 {m\sin(\delta_{\mathcal A}/2)}.
 \label{eq:G_close_identity}
\end{equation}
The condition
\begin{equation}
 m
 \ge
 \frac{2(r-1)}
 {\sin(\delta_{\mathcal A}/2)}
\end{equation}
therefore ensures
\begin{equation}
 \|G-I\|_1\le\frac12.
\end{equation}
Hence $G$ is invertible.  Indeed, writing
\begin{equation}
 G=I+(G-I),
\end{equation}
the Neumann series converges and yields
\begin{equation}
 G^{-1}
 =
 \sum_{k=0}^{\infty}
 (-1)^k(G-I)^k,
\end{equation}
so that
\begin{equation}
 \|G^{-1}\|_1
 \le
 \frac1{1-\|G-I\|_1}
 \le
 2.
 \label{eq:G_inverse}
\end{equation}

Now define the target inverse values on the active spectrum by
\begin{equation}
 d_a:=\frac1{\lambda_a},
 \qquad
 a\in\mathcal A,
\end{equation}
and collect them in the vector $\mathbf d\in\mathbb C^r$.  Define
\begin{equation}
 \boldsymbol\alpha
 :=
 G^{-1}\mathbf d.
 \label{eq:alpha_definition}
\end{equation}
Using these coefficients, introduce
\begin{equation}
 q_m(\lambda)
 :=
 \sum_{a\in\mathcal A}
 \alpha_a F_{a,m}(\lambda).
 \label{eq:q_interpolant}
\end{equation}

For any active mode $b\in\mathcal A$,
\begin{align}
 q_m(\lambda_b)
 &=
 \sum_{a\in\mathcal A}
 \alpha_aF_{a,m}(\lambda_b)
 \nonumber\\
 &=
 (G\boldsymbol\alpha)_b
 =
 d_b
 =
 \frac1{\lambda_b}.
 \label{eq:q_exact_active}
\end{align}
Thus $q_m$ reproduces $1/\lambda$ exactly on the complete active spectrum.

\paragraph{Step 3: Construction of an explicit vector in the QKS subspace.}

Define
\begin{equation}
 z_m
 :=
 q_m(A)|b\rangle.
 \label{eq:z_construction}
\end{equation}
Each localized filter can be written as
\begin{align}
 F_{a,m}(\lambda)
 &=
 \frac1m
 \sum_{j=0}^{m-1}
 e^{-ij(\tau\lambda-\theta_a)}
 \nonumber\\
 &=
 \sum_{j=0}^{m-1}
 \frac{e^{ij\theta_a}}{m}
 e^{-ij\tau\lambda}.
\end{align}
Hence $q_m(\lambda)$ is itself a linear combination of
\begin{equation}
 1,\,
 e^{-i\tau\lambda},\,
 \ldots,\,
 e^{-i(m-1)\tau\lambda}.
\end{equation}
Therefore there exist coefficients $\gamma_0,\ldots,\gamma_{m-1}$ such
that
\begin{equation}
 q_m(A)|b\rangle
 =
 \sum_{j=0}^{m-1}
 \gamma_j e^{-ij\tau A}|b\rangle.
\end{equation}
By definition of the QKS subspace,
\begin{equation}
 \boxed{
 z_m\in\mathcal K_m.}
 \label{eq:z_in_krylov}
\end{equation}

Furthermore, expanding in the eigenbasis,
\begin{equation}
 z_m
 =
 \sum_{\ell=1}^{N}
 b_\ell q_m(\lambda_\ell)|u_\ell\rangle.
 \label{eq:z_spectral_expansion}
\end{equation}
For every active mode $a\in\mathcal A$,
Eq.~\eqref{eq:q_exact_active} gives
\begin{equation}
 b_a q_m(\lambda_a)
 =
 \frac{b_a}{\lambda_a}.
\end{equation}
Since
\begin{equation}
 x
 =
 \sum_{\ell=1}^{N}
 \frac{b_\ell}{\lambda_\ell}|u_\ell\rangle,
\end{equation}
the constructed state therefore reproduces the exact solution on the
active spectral subspace:
\begin{equation}
 P_{\mathcal A}z_m
 =
 P_{\mathcal A}x,
 \label{eq:exact_active_response}
\end{equation}
where
\begin{equation}
 P_{\mathcal A}
 :=
 \sum_{a\in\mathcal A}
 |u_a\rangle\langle u_a|
\end{equation}
is the orthogonal projector onto the active eigenspace.

\paragraph{Step 4: Bounding the interpolation coefficients.}

To control the inactive components of $z_m$, we first bound
$\|\boldsymbol\alpha\|_1$.

For every active mode $a\in\mathcal A$, orthogonality of the eigenvectors
implies
\begin{equation}
 \|x\|_2
 \ge
 \frac{|b_a|}{\lambda_a}.
\end{equation}
Using the active-mode overlap condition
\begin{equation}
 |b_a|\ge\beta_{\mathcal A},
\end{equation}
we obtain
\begin{equation}
 \frac1{\lambda_a}
 \le
 \frac{\|x\|_2}{\beta_{\mathcal A}}.
 \label{eq:inverse_active_bound}
\end{equation}
Since there are $r$ active modes,
\begin{equation}
 \|\mathbf d\|_1
 =
 \sum_{a\in\mathcal A}\frac1{\lambda_a}
 \le
 \frac{r}{\beta_{\mathcal A}}\|x\|_2.
 \label{eq:d_bound}
\end{equation}
Combining Eq.~\eqref{eq:d_bound} with
$\boldsymbol\alpha=G^{-1}\mathbf d$ and
Eq.~\eqref{eq:G_inverse} gives
\begin{equation}
 \boxed{
 \|\boldsymbol\alpha\|_1
 \le
 \frac{2r}{\beta_{\mathcal A}}
 \|x\|_2.}
 \label{eq:alpha_bound}
\end{equation}

\paragraph{Step 5: Bounding the error on the inactive spectrum.}

We partition the inactive set into two parts.  The unresolved inactive set
$\mathcal U_\rho$ contains modes whose phases lie within distance $\rho$ of
at least one active phase.  Define the complementary resolved inactive set
by
\begin{equation}
 \mathcal I_{\rho}^{\rm res}
 :=
 \mathcal I\setminus\mathcal U_\rho.
 \label{eq:resolved_inactive_set}
\end{equation}

For any $\ell\in\mathcal I_{\rho}^{\rm res}$ and any
$a\in\mathcal A$,
\begin{equation}
 d_{\mathbb T}(\theta_\ell,\theta_a)
 \ge
 \rho.
\end{equation}
Lemma~\ref{lem:localized_filter} therefore gives
\begin{equation}
 |F_{a,m}(\lambda_\ell)|
 \le
 \frac1{m\sin(\rho/2)}.
\end{equation}
Using Eq.~\eqref{eq:q_interpolant},
\begin{align}
 |q_m(\lambda_\ell)|
 &\le
 \sum_{a\in\mathcal A}
 |\alpha_a|
 |F_{a,m}(\lambda_\ell)|
 \nonumber\\
 &\le
 \frac{\|\boldsymbol\alpha\|_1}
 {m\sin(\rho/2)}.
 \label{eq:q_resolved_bound}
\end{align}

Let
\begin{equation}
 P_{\mathcal I_{\rho}^{\rm res}}
 :=
 \sum_{\ell\in\mathcal I_{\rho}^{\rm res}}
 |u_\ell\rangle\langle u_\ell|.
\end{equation}
Then, using Eq.~\eqref{eq:z_spectral_expansion},
\begin{align}
 \|P_{\mathcal I_{\rho}^{\rm res}}z_m\|_2^2
 &=
 \sum_{\ell\in\mathcal I_{\rho}^{\rm res}}
 |b_\ell|^2
 |q_m(\lambda_\ell)|^2
 \nonumber\\
 &\le
 \frac{\|\boldsymbol\alpha\|_1^2}
 {m^2\sin^2(\rho/2)}
 \sum_{\ell\in\mathcal I_{\rho}^{\rm res}}
 |b_\ell|^2
 \nonumber\\
 &\le
 \frac{\|\boldsymbol\alpha\|_1^2}
 {m^2\sin^2(\rho/2)}.
\end{align}
Using Eq.~\eqref{eq:alpha_bound},
\begin{equation}
 \boxed{
 \frac{
 \|P_{\mathcal I_{\rho}^{\rm res}}z_m\|_2
 }{
 \|x\|_2
 }
 \le
 \frac{2r}{\beta_{\mathcal A}}
 \frac1{m\sin(\rho/2)}.}
 \label{eq:resolved_inactive}
\end{equation}

For the unresolved inactive modes, phase separation cannot be used to
obtain an $O(1/m)$ suppression.  Instead, the trivial bound
\begin{equation}
 |F_{a,m}(\lambda)|
 \le
 \frac1m
 \sum_{j=0}^{m-1}1
 =
 1
\end{equation}
gives
\begin{equation}
 |q_m(\lambda_\ell)|
 \le
 \|\boldsymbol\alpha\|_1,
 \qquad
 \ell\in\mathcal U_\rho.
\end{equation}
Let
\begin{equation}
 P_{\mathcal U_\rho}
 :=
 \sum_{\ell\in\mathcal U_\rho}
 |u_\ell\rangle\langle u_\ell|.
\end{equation}
Then
\begin{align}
 \|P_{\mathcal U_\rho}z_m\|_2^2
 &=
 \sum_{\ell\in\mathcal U_\rho}
 |b_\ell|^2
 |q_m(\lambda_\ell)|^2
 \nonumber\\
 &\le
 \|\boldsymbol\alpha\|_1^2
 \sum_{\ell\in\mathcal U_\rho}|b_\ell|^2
 \nonumber\\
 &=
 \|\boldsymbol\alpha\|_1^2
 \epsilon_{\rm U}^2.
\end{align}
Therefore
\begin{equation}
 \boxed{
 \frac{
 \|P_{\mathcal U_\rho}z_m\|_2
 }{
 \|x\|_2
 }
 \le
 \frac{2r}{\beta_{\mathcal A}}
 \epsilon_{\rm U}.}
 \label{eq:unresolved_inactive}
\end{equation}

Because $x$ and $z_m$ agree exactly on the active subspace,
Eq.~\eqref{eq:exact_active_response} gives
\begin{equation}
 x-z_m
 =
 x_{\mathcal I}
 -
 P_{\mathcal U_\rho}z_m
 -
 P_{\mathcal I_{\rho}^{\rm res}}z_m,
\end{equation}
where
\begin{equation}
 x_{\mathcal I}
 :=
 P_{\mathcal I}x.
\end{equation}
Applying the triangle inequality,
\begin{align}
 \frac{\|x-z_m\|_2}{\|x\|_2}
 \le{}&
 \frac{\|x_{\mathcal I}\|_2}{\|x\|_2}
 +
 \frac{\|P_{\mathcal U_\rho}z_m\|_2}{\|x\|_2}
 \nonumber\\
 &+
 \frac{
 \|P_{\mathcal I_{\rho}^{\rm res}}z_m\|_2
 }{
 \|x\|_2
 }.
\end{align}
Assumption (A1) gives
\begin{equation}
 \frac{\|x_{\mathcal I}\|_2}{\|x\|_2}
 \le
 \epsilon_{\rm sub}.
\end{equation}
Combining this with
Eqs.~\eqref{eq:resolved_inactive} and
\eqref{eq:unresolved_inactive} yields
\begin{equation}
 \boxed{
 \frac{\|x-z_m\|_2}{\|x\|_2}
 \le
 \epsilon_{\rm sub}
 +
 \frac{2r}{\beta_{\mathcal A}}
 \left[
 \epsilon_{\rm U}
 +
 \frac1{m\sin(\rho/2)}
 \right].}
 \label{eq:structural_error_proof}
\end{equation}

For a prescribed target error $\epsilon_K$, define the remaining error
budget
\begin{equation}
 \Delta_K
 =
 \epsilon_K
 -
 \epsilon_{\rm sub}
 -
 \frac{2r}{\beta_{\mathcal A}}\epsilon_{\rm U}.
\end{equation}
If $\Delta_K>0$, requiring
\begin{equation}
 \frac{2r}
 {\beta_{\mathcal A}m\sin(\rho/2)}
 \le
 \Delta_K
\end{equation}
is sufficient to guarantee
\begin{equation}
 \frac{\|x-z_m\|_2}{\|x\|_2}
 \le
 \epsilon_K.
\end{equation}
Equivalently,
\begin{equation}
 m
 \ge
 \frac{2r}
 {\beta_{\mathcal A}\Delta_K\sin(\rho/2)}.
 \label{eq:m_accuracy_requirement_proof}
\end{equation}
Combining Eq.~\eqref{eq:m_accuracy_requirement_proof} with the active-mode
interpolation requirement derived from Eq.~\eqref{eq:G_close_identity}
gives the sufficient Krylov-dimension requirement stated in
Eq.~\eqref{eq:structural_m_requirement}.

\paragraph{Step 6: Stability of the Krylov Gram matrix and state reconstruction.}

Define the full phase matrix $V\in\mathbb C^{N\times m}$ by
\begin{equation}
 V_{\ell j}
 :=
 e^{-ij\tau\lambda_\ell},
 \qquad
 \ell=1,\ldots,N,
 \quad
 j=0,\ldots,m-1,
\end{equation}
and the spectral weight matrix by
\begin{equation}
 W
 :=
 \operatorname{diag}
 \left(
 |b_1|^2,\ldots,|b_N|^2
 \right).
\end{equation}
The QKS Gram matrix can then be written as
\begin{equation}
 S_m
 =
 \Phi_m^\dagger\Phi_m
 =
 V^\dagger W V.
 \label{eq:S_weighted_V}
\end{equation}

Assumption (A2) states that there exists a spectral support set
$\mathcal R\subseteq\{1,\ldots,N\}$ with $|\mathcal R|\ge m$ such that
\begin{equation}
 \sigma_{\min}
 \left(
 W_{\mathcal R}^{1/2}
 V_{\mathcal R}
 \right)
 \ge
 s_0>0.
 \label{eq:support_stability_recall}
\end{equation}
Since every omitted spectral contribution is positive semidefinite,
\begin{equation}
 S_m
 =
 V^\dagger WV
 \succeq
 V_{\mathcal R}^\dagger
 W_{\mathcal R}
 V_{\mathcal R}.
 \label{eq:S_support_order}
\end{equation}
Moreover,
\begin{align}
 \lambda_{\min}
 \left(
 V_{\mathcal R}^\dagger
 W_{\mathcal R}
 V_{\mathcal R}
 \right)
 &=
 \sigma_{\min}^2
 \left(
 W_{\mathcal R}^{1/2}
 V_{\mathcal R}
 \right)
 \nonumber\\
 &\ge
 s_0^2.
\end{align}
Therefore
\begin{equation}
 \boxed{
 \lambda_{\min}(S_m)\ge s_0^2.}
\end{equation}

Finally, the LCU reconstruction success probability for any nonzero
coefficient vector $\boldsymbol\gamma$ is
\begin{equation}
 p_{\rm succ}
 =
 \frac{
 \boldsymbol\gamma^\dagger
 S_m
 \boldsymbol\gamma
 }{
 \|\boldsymbol\gamma\|_1^2
 }.
\end{equation}
The Gram lower bound gives
\begin{equation}
 \boldsymbol\gamma^\dagger
 S_m
 \boldsymbol\gamma
 \ge
 s_0^2
 \|\boldsymbol\gamma\|_2^2,
\end{equation}
while Cauchy--Schwarz gives
\begin{equation}
 \|\boldsymbol\gamma\|_1^2
 \le
 m\|\boldsymbol\gamma\|_2^2.
\end{equation}
Hence
\begin{equation}
 \boxed{
 p_{\rm succ}
 \ge
 \frac{s_0^2}{m}.}
\end{equation}

Combining the interpolation and approximation requirements yields the
explicit structural scaling
\begin{equation}
 m
 =
 O\!\left(
 \max\left\{
 \frac{r}{\sin(\delta_{\mathcal A}/2)},
 \,
 \frac{r}
 {\beta_{\mathcal A}\Delta_K\sin(\rho/2)}
 \right\}
 \right),
\end{equation}
with
\begin{equation}
 T=(m-1)\tau,
 \qquad
 p_{\rm succ}^{-1}
 \le
 \frac{m}{s_0^2}.
\end{equation}

Therefore, if the solution-weighted spectral parameters
$r$, $\delta_{\mathcal A}$, $\rho$, $\beta_{\mathcal A}$,
$\Delta_K$, $s_0$, and $\tau$
remain bounded independently of the vanishing solution-relevant
eigenvalue, none of
$m$, $T$, or $p_{\rm succ}^{-1}$
inherits its divergence.
\end{proof}

\section{Proofs for Directional Stability and Galerkin Transfer}
\label{app:directional_proofs}

\subsection{Stable transfer to orthonormal reduced coordinates}
\label{app:orthonormalization_proof}

\begin{proof}[Proof of Theorem~\ref{thm:SHg_Bc}]
The Gram matrix $S_m$ is used here only as a fixed theoretical coordinate
transform.  Since
\begin{equation}
    B_m=S_m^{-1/2}H_mS_m^{-1/2},
    \qquad
    \mathbf c_m=S_m^{-1/2}\mathbf g_m,
\end{equation}
the perturbed analysis coordinates associated with
$\widehat H_m=H_m+\Delta H_m$ and
$\widehat{\mathbf g}_m=\mathbf g_m+\Delta\mathbf g_m$ satisfy
\begin{equation}
    \widehat B_m-B_m
    =
    S_m^{-1/2}\Delta H_mS_m^{-1/2},
\end{equation}
and
\begin{equation}
    \widehat{\mathbf c}_m-\mathbf c_m
    =
    S_m^{-1/2}\Delta\mathbf g_m.
\end{equation}
The assumption
$\lambda_{\min}(S_m)\ge\sigma_S>0$
implies
\begin{equation}
    \|S_m^{-1/2}\|_2
    \le
    \sigma_S^{-1/2}.
\end{equation}
Therefore
\begin{equation}
    \|\widehat B_m-B_m\|_2
    \le
    \frac{\|\Delta H_m\|_2}{\sigma_S},
\end{equation}
and
\begin{equation}
    \|\widehat{\mathbf c}_m-\mathbf c_m\|_2
    \le
    \frac{\|\Delta\mathbf g_m\|_2}{\sqrt{\sigma_S}}.
\end{equation}
Thus the change to orthonormal coordinates introduces only fixed
Gram-geometry factors and no inverse dependence on the vanishing reduced
eigenvalue.
\end{proof}

\subsection{Ideal directional stability}
\label{app:solution_aligned_proof}

\begin{proof}[Proof of Theorem~\ref{thm:solution_aligned}]

The proof shows that the only divergent component of $B_m^{-1}$ becomes
asymptotically parallel to the solution itself and is therefore removed by
the transverse projector entering $K_B$.

\paragraph{Step 1: Alignment of the reduced solution with the soft direction.}

Under (B1)--(B2), write
\begin{equation}
    \mathbf c_m
    =
    c_1\mathbf v_1
    +
    \mathbf c_\perp,
    \qquad
    \mathbf c_\perp
    =
    \sum_{i\ge2}c_i\mathbf v_i.
\end{equation}
The exact reduced solution is
\begin{equation}
    \mathbf y
    =
    B_m^{-1}\mathbf c_m
    =
    \frac{c_1}{\mu_1}\mathbf v_1
    +
    \sum_{i\ge2}
    \frac{c_i}{\mu_i}\mathbf v_i.
    \label{eq:appendix_y_decomposition}
\end{equation}
Define the transverse component
\begin{equation}
    \mathbf y_\perp
    :=
    \sum_{i\ge2}
    \frac{c_i}{\mu_i}\mathbf v_i.
\end{equation}
Since $\mu_i\ge\mu_2$ for all $i\ge2$,
\begin{equation}
    \|\mathbf y_\perp\|_2
    \le
    \frac{\|\mathbf c_\perp\|_2}{\mu_2}.
\end{equation}
Meanwhile,
\begin{equation}
    \|\mathbf y\|_2
    \ge
    \frac{|c_1|}{\mu_1}.
\end{equation}
If
\begin{equation}
    \vartheta
    :=
    \Theta(\mathbf y,\mathbf v_1),
\end{equation}
then
\begin{align}
    \sin\vartheta
    &=
    \frac{\|\mathbf y_\perp\|_2}{\|\mathbf y\|_2}
    \nonumber\\
    &\le
    \frac{\mu_1}{\mu_2}
    \frac{\|\mathbf c_\perp\|_2}{|c_1|}.
    \label{eq:appendix_alignment_angle}
\end{align}

Thus the normalized solution becomes increasingly aligned with
$\mathbf v_1$ as $\mu_1\to0$.

\paragraph{Step 2: Projection removes the divergent inverse direction.}

The reduced inverse has the spectral representation
\begin{equation}
    B_m^{-1}
    =
    \frac1{\mu_1}
    \mathbf v_1\mathbf v_1^\dagger
    +
    \sum_{i\ge2}
    \frac1{\mu_i}
    \mathbf v_i\mathbf v_i^\dagger.
\end{equation}
Applying the transverse projector
$P_{\mathbf y}^{\perp}$ gives
\begin{align}
    K_B
    &=
    \|P_{\mathbf y}^{\perp}B_m^{-1}\|_2
    \nonumber\\
    &\le
    \frac{
    \|P_{\mathbf y}^{\perp}\mathbf v_1\|_2
    }{\mu_1}
    +
    \frac1{\mu_2}.
\end{align}
Because
\begin{equation}
    \|P_{\mathbf y}^{\perp}\mathbf v_1\|_2
    =
    \sin\vartheta,
\end{equation}
Eq.~\eqref{eq:appendix_alignment_angle} yields
\begin{equation}
    \boxed{
    K_B
    \le
    \frac1{\mu_2}
    \left(
    1+
    \frac{\|\mathbf c_\perp\|_2}{|c_1|}
    \right).
    }
\end{equation}

Under (B1)--(B2),
$\mu_2\ge\Delta_B>0$ and
$|c_1|\ge\beta_B>0$.
Moreover,
$\|\mathbf c_\perp\|_2\le\|\mathbf c_m\|_2\le1$ because
$\mathbf c_m=Q_m^\dagger|b\rangle$ and $Q_m$ is an isometry.
Hence
\begin{equation}
    K_B=O(1)
\end{equation}
with respect to $\mu_1$.

\paragraph{Step 3: Right-hand-side sensitivity decreases with the soft eigenvalue.}

By definition,
\begin{equation}
    K_c
    =
    \frac{K_B}{\|\mathbf y\|_2}.
\end{equation}
Using
\begin{equation}
    \|\mathbf y\|_2
    \ge
    \frac{|c_1|}{\mu_1},
\end{equation}
we obtain
\begin{equation}
    \boxed{
    K_c
    \le
    \frac{\mu_1}{|c_1|\mu_2}
    \left(
    1+
    \frac{\|\mathbf c_\perp\|_2}{|c_1|}
    \right).
    }
\end{equation}
Therefore
\begin{equation}
    K_c=O(\mu_1).
\end{equation}

The full inverse norm diverges as
$\|B_m^{-1}\|_2=\mu_1^{-1}$, but this divergence is asymptotically aligned
with the solution direction.  The transverse normalized-state sensitivity
therefore remains bounded.
\end{proof}

\subsection{Finite-error directional stability}
\label{app:finite_stability_proof}

\begin{proof}[Proof of Theorem~\ref{thm:finite_stability}]

The theorem is stated in the main text using universal constants because
their precise numerical values are not essential.  Here we give one
explicit sufficient choice to establish their existence.

For concreteness, assume
\begin{equation}
    \mu_1\le\frac{\Delta_B}{4},
    \qquad
    \eta_B\le\frac{\Delta_B}{8}.
    \label{eq:appendix_concrete_gap_conditions}
\end{equation}
These numerical fractions are not intrinsic to the result; they only
provide a convenient fixed separation between the soft eigenvalue, the
remaining spectrum, and the perturbation scale.

\paragraph{Step 1: The isolated soft eigendirection remains stable.}

Since
\begin{equation}
    \mu_2\ge\Delta_B,
\end{equation}
Eq.~\eqref{eq:appendix_concrete_gap_conditions} gives
\begin{equation}
    \mu_2-\mu_1
    \ge
    \frac{3\Delta_B}{4}.
\end{equation}
Weyl's inequality gives
\begin{equation}
    |\widehat\mu_i-\mu_i|
    \le
    \eta_B.
    \label{eq:appendix_weyl}
\end{equation}
In particular, for all $i\ge2$,
\begin{equation}
    \widehat\mu_i
    \ge
    \Delta_B-\eta_B
    \ge
    \frac{7\Delta_B}{8}.
    \label{eq:appendix_transverse_gap}
\end{equation}

The fixed gap therefore remains open after perturbation.  By the
Davis--Kahan theorem,
\begin{equation}
    \sin\Theta(\mathbf v_1,\widehat{\mathbf v}_1)
    \le
    \frac{\eta_B}
    {\mu_2-\mu_1-\eta_B}
    \le
    \frac{2\eta_B}{\Delta_B}.
    \label{eq:appendix_DK_bound}
\end{equation}
Define
\begin{equation}
    \delta_v
    :=
    \frac{2\eta_B}{\Delta_B}.
\end{equation}

Thus the finite operator error $\eta_B$ rotates the soft eigendirection by
an amount controlled by $\eta_B/\Delta_B$, rather than by
$\eta_B/\mu_1$.

\paragraph{Step 2: The perturbed right-hand side retains soft-mode overlap.}

Choose the phase of $\widehat{\mathbf v}_1$ so that its overlap with
$\mathbf v_1$ is nonnegative.  For unit vectors,
\begin{equation}
    \|\widehat{\mathbf v}_1-\mathbf v_1\|_2
    \le
    \sqrt{2}\,
    \sin\Theta(\mathbf v_1,\widehat{\mathbf v}_1)
    \le
    \sqrt{2}\,\delta_v.
    \label{eq:appendix_eigenvector_difference}
\end{equation}

Because
\begin{equation}
    \widehat{\mathbf c}_m
    =
    \mathbf c_m+\mathbf f,
\end{equation}
we have
\begin{align}
    |\langle
    \widehat{\mathbf v}_1,
    \widehat{\mathbf c}_m
    \rangle|
    \ge{}&
    |\langle\mathbf v_1,\mathbf c_m\rangle|
    -
    \|\widehat{\mathbf v}_1-\mathbf v_1\|_2
    \|\mathbf c_m\|_2
    -
    \|\mathbf f\|_2.
\end{align}
Using
$|c_1|\ge\beta_B$,
$\|\mathbf c_m\|_2\le1$,
and
$\|\mathbf f\|_2\le\eta_c$,
\begin{equation}
    |\langle
    \widehat{\mathbf v}_1,
    \widehat{\mathbf c}_m
    \rangle|
    \ge
    \beta_B
    -
    \sqrt2\,\delta_v
    -
    \eta_c.
\end{equation}

Therefore the sufficient precision requirement
\begin{equation}
    \eta_c+\sqrt2\,\delta_v
    \le
    \frac{\beta_B}{2}
    \label{eq:appendix_overlap_requirement}
\end{equation}
guarantees
\begin{equation}
    \boxed{
    |\langle
    \widehat{\mathbf v}_1,
    \widehat{\mathbf c}_m
    \rangle|
    \ge
    \frac{\beta_B}{2}.}
    \label{eq:appendix_perturbed_overlap}
\end{equation}

Equation~\eqref{eq:appendix_overlap_requirement} is therefore a sufficient
measurement-precision condition: the combined operator and
right-hand-side errors must not be large enough to destroy the original
soft-mode overlap.  Crucially, this condition is controlled by
$\Delta_B$ and $\beta_B$, not by $\mu_1$.

\paragraph{Step 3: The exact solution remains aligned with the original soft direction.}

From the proof of Theorem~\ref{thm:solution_aligned},
\begin{equation}
    \sin\Theta(\mathbf y,\mathbf v_1)
    \le
    \frac{\mu_1}{\mu_2}
    \frac{\|\mathbf c_\perp\|_2}{|c_1|}.
\end{equation}
Using
$\mu_2\ge\Delta_B$,
$|c_1|\ge\beta_B$,
and
$\|\mathbf c_\perp\|_2\le1$,
\begin{equation}
    \boxed{
    \sin\Theta(\mathbf y,\mathbf v_1)
    \le
    \frac{\mu_1}{\beta_B\Delta_B}.}
    \label{eq:appendix_exact_alignment}
\end{equation}

\paragraph{Step 4: The perturbed transverse response remains bounded.}

Expand the perturbed right-hand side in the eigenbasis of
$\widehat B_m$:
\begin{equation}
    \widehat{\mathbf c}_m
    =
    \widehat c_1\widehat{\mathbf v}_1
    +
    \widehat{\mathbf c}_\perp.
\end{equation}
The perturbed solution is
\begin{equation}
    \widehat{\mathbf y}
    =
    \frac{\widehat c_1}{\widehat\mu_1}
    \widehat{\mathbf v}_1
    +
    \widehat{\mathbf y}_\perp,
\end{equation}
where
\begin{equation}
    \widehat{\mathbf y}_\perp
    =
    \sum_{i\ge2}
    \frac{\widehat c_i}{\widehat\mu_i}
    \widehat{\mathbf v}_i.
\end{equation}

By Eq.~\eqref{eq:appendix_transverse_gap},
\begin{align}
    \|\widehat{\mathbf y}_\perp\|_2
    &\le
    \frac{8}{7\Delta_B}
    \|\widehat{\mathbf c}_\perp\|_2
    \nonumber\\
    &\le
    \frac{8(1+\eta_c)}
    {7\Delta_B},
    \label{eq:appendix_transverse_response}
\end{align}
where we used
\begin{equation}
    \|\widehat{\mathbf c}_m\|_2
    \le
    \|\mathbf c_m\|_2+\|\mathbf f\|_2
    \le
    1+\eta_c.
\end{equation}

Thus finite perturbations do not produce a divergent response in any
direction transverse to the isolated soft mode.

\paragraph{Step 5: The perturbed solution remains aligned with the perturbed soft direction.}

Weyl's inequality also gives
\begin{equation}
    |\widehat\mu_1|
    \le
    \mu_1+\eta_B.
\end{equation}
Together with
Eq.~\eqref{eq:appendix_perturbed_overlap},
\begin{equation}
    \|\widehat{\mathbf y}\|_2
    \ge
    \frac{|\widehat c_1|}
    {|\widehat\mu_1|}
    \ge
    \frac{\beta_B}
    {2(\mu_1+\eta_B)}.
    \label{eq:appendix_perturbed_solution_lower}
\end{equation}

Combining
Eqs.~\eqref{eq:appendix_transverse_response}
and~\eqref{eq:appendix_perturbed_solution_lower},
\begin{align}
    \sin\Theta(
    \widehat{\mathbf y},
    \widehat{\mathbf v}_1)
    &=
    \frac{\|\widehat{\mathbf y}_\perp\|_2}
    {\|\widehat{\mathbf y}\|_2}
    \nonumber\\
    &\le
    \frac{
    16(1+\eta_c)(\mu_1+\eta_B)
    }{
    7\beta_B\Delta_B
    }.
    \label{eq:appendix_perturbed_alignment}
\end{align}

Thus the perturbed solution remains aligned with the perturbed soft
direction.  Again, no inverse power of $\mu_1$ appears.

\paragraph{Step 6: Combine the three directional deviations.}

The phase-invariant angle satisfies the triangle inequality:
\begin{equation}
    \Theta(
    \mathbf y,
    \widehat{\mathbf y})
    \le
    \Theta(\mathbf y,\mathbf v_1)
    +
    \Theta(\mathbf v_1,\widehat{\mathbf v}_1)
    +
    \Theta(
    \widehat{\mathbf v}_1,
    \widehat{\mathbf y}).
    \label{eq:appendix_angle_chain}
\end{equation}
Using
Eqs.~\eqref{eq:appendix_exact_alignment},
\eqref{eq:appendix_DK_bound},
and~\eqref{eq:appendix_perturbed_alignment},
we obtain the explicit bound
\begin{align}
    \Theta(
    \mathbf y,
    \widehat{\mathbf y})
    \le{}&
    \arcsin\!\left(
    \frac{\mu_1}{\beta_B\Delta_B}
    \right)
    +
    \arcsin\!\left(
    \frac{2\eta_B}{\Delta_B}
    \right)
    \nonumber\\
    &+
    \arcsin\!\left[
    \frac{
    16(1+\eta_c)(\mu_1+\eta_B)
    }{
    7\beta_B\Delta_B
    }
    \right].
    \label{eq:appendix_finite_exact_bound}
\end{align}

In any fixed perturbative regime in which the three inverse-sine
arguments are bounded away from one, the inequality
$\arcsin t\le C t$ for a universal constant $C$ yields
\begin{equation}
    \Theta(
    \mathbf y,
    \widehat{\mathbf y})
    \le
    C_{\rm stab}
    \left[
    \frac{\mu_1}{\beta_B\Delta_B}
    +
    \frac{\eta_B}{\Delta_B}
    +
    \frac{(1+\eta_c)(\mu_1+\eta_B)}
    {\beta_B\Delta_B}
    \right].
\end{equation}

This is the abstract bound stated in
Theorem~\ref{thm:finite_stability}.  One admissible choice of the abstract
constants is therefore obtained from
Eqs.~\eqref{eq:appendix_concrete_gap_conditions}
and~\eqref{eq:appendix_overlap_requirement}.

The key conclusion is that the permitted finite error radius is controlled
by the fixed spectral gap $\Delta_B$ and the persistent overlap
$\beta_B$.  It does not have to shrink proportionally to the vanishing
eigenvalue $\mu_1$.
\end{proof}

\subsection{Near-singular estimated soft modes}
\label{app:near_singular_regularization}

The finite-stability theorem assumes that $\widehat B_m$ is nonsingular so
that $\widehat{\mathbf y}=\widehat B_m^{-1}\widehat{\mathbf c}_m$ is
well defined.  This subsection treats the finite-precision edge case in
which the isolated estimated soft eigenvalue is zero or numerically
indistinguishable from zero.

Let
\begin{equation}
    \widehat B_m
    =
    \widehat\mu_1
    \widehat{\mathbf v}_1
    \widehat{\mathbf v}_1^\dagger
    +
    \sum_{i\ge2}
    \widehat\mu_i
    \widehat{\mathbf v}_i
    \widehat{\mathbf v}_i^\dagger,
\end{equation}
and suppose the finite-error conditions of
Theorem~\ref{thm:finite_stability} hold, so that
\begin{equation}
    |\widehat c_1|
    :=
    |\langle
    \widehat{\mathbf v}_1,
    \widehat{\mathbf c}_m
    \rangle|
    \ge
    \frac{\beta_B}{2}
\end{equation}
and the transverse spectrum remains bounded away from zero.

If
\begin{equation}
    |\widehat\mu_1|<\nu
\end{equation}
for a chosen numerical floor $\nu>0$, replace only the isolated soft
eigenvalue by a nonzero value $\widetilde\mu_1$ with
\begin{equation}
    |\widetilde\mu_1|=\nu,
\end{equation}
preserving the sign of $\widehat\mu_1$ whenever it is nonzero (and taking the positive sign in the exactly singular positive-definite limit)
while leaving all eigenvectors and all remaining eigenvalues unchanged.
Define
\begin{equation}
    \widetilde B_m
    :=
    \widetilde\mu_1
    \widehat{\mathbf v}_1
    \widehat{\mathbf v}_1^\dagger
    +
    \sum_{i\ge2}
    \widehat\mu_i
    \widehat{\mathbf v}_i
    \widehat{\mathbf v}_i^\dagger
\end{equation}
and
\begin{equation}
    \widetilde{\mathbf y}
    :=
    \widetilde B_m^{-1}
    \widehat{\mathbf c}_m.
\end{equation}

Decompose
\begin{equation}
    \widetilde{\mathbf y}
    =
    \frac{\widehat c_1}{\widetilde\mu_1}
    \widehat{\mathbf v}_1
    +
    \widetilde{\mathbf y}_\perp.
\end{equation}
The transverse component is unchanged from the finite-error analysis and
satisfies
\begin{equation}
    \|\widetilde{\mathbf y}_\perp\|_2
    \le
    \frac{8(1+\eta_c)}
    {7\Delta_B}.
\end{equation}
Meanwhile,
\begin{equation}
    \|\widetilde{\mathbf y}\|_2
    \ge
    \frac{|\widehat c_1|}{\nu}
    \ge
    \frac{\beta_B}{2\nu}.
\end{equation}
Therefore
\begin{equation}
    \boxed{
    \sin\Theta(
    \widetilde{\mathbf y},
    \widehat{\mathbf v}_1)
    \le
    \frac{
    16(1+\eta_c)\nu
    }{
    7\beta_B\Delta_B
    }.}
    \label{eq:regularized_soft_mode_angle}
\end{equation}

Thus the additional normalized-direction error introduced by the spectral
floor is proportional to $\nu$.  Choosing
\begin{equation}
    \nu
    =
    O\!\left(
    \frac{\beta_B\Delta_B}{1+\eta_c}
    \epsilon
    \right)
\end{equation}
is sufficient to make this contribution $O(\epsilon)$.  In the favorable
regime where $\beta_B,\Delta_B=\Omega(1)$, this reduces to
\begin{equation}
    \nu=O(\epsilon).
\end{equation}

Crucially, the required numerical floor does not scale with the physical
soft eigenvalue $\mu_1$.  The procedure therefore regularizes only an
exactly or numerically singular \emph{estimate}; it does not truncate the
solution-relevant physical soft mode.

\subsection{Galerkin transfer}
\label{app:galerkin_transfer_proof}

\begin{proof}[Proof of Theorem~\ref{thm:galerkin_transfer}]

The proof transfers the existence of an accurate vector in the Krylov
subspace to the accuracy of the Galerkin solution actually returned by
QKS.  The argument proceeds by comparing the two vectors within the same
subspace and controlling only the component that changes the normalized
solution direction.

\paragraph{Step 1: Represent the candidate and Galerkin solution in the same subspace.}

Both the candidate vector $z$ and the QKS Galerkin solution $x_m$ belong
to the same Krylov subspace $\mathcal K_m$.  For the analysis only, let
$Q_m$ be an orthonormal basis for $\mathcal K_m$.  Then there exist
coordinate vectors $\mathbf y_z$ and $\mathbf y$ such that
\begin{equation}
    z=Q_m\mathbf y_z,
    \qquad
    x_m=Q_m\mathbf y.
    \label{eq:galerkin_same_subspace_coordinates}
\end{equation}
Because $Q_m^\dagger Q_m=I$, this change of coordinates preserves both
Euclidean norms and angles:
\begin{equation}
    \|\mathbf y_z\|_2=\|z\|_2,
    \qquad
    \Theta(\mathbf y,\mathbf y_z)
    =
    \Theta(x_m,z).
    \label{eq:galerkin_isometry}
\end{equation}

\paragraph{Step 2: Relate the Galerkin solution to the accurate candidate.}

The orthonormal reduced operator is
\begin{equation}
    B_m=Q_m^\dagger A Q_m.
\end{equation}
The Galerkin solution satisfies
\begin{equation}
    B_m\mathbf y
    =
    Q_m^\dagger |b\rangle.
\end{equation}
Since the exact solution obeys
\begin{equation}
    A x=|b\rangle,
\end{equation}
we may equivalently write
\begin{equation}
    B_m\mathbf y
    =
    Q_m^\dagger A x.
    \label{eq:galerkin_exact_rhs}
\end{equation}

On the other hand, because $z=Q_m\mathbf y_z$,
\begin{equation}
    B_m\mathbf y_z
    =
    Q_m^\dagger A Q_m\mathbf y_z
    =
    Q_m^\dagger A z.
    \label{eq:galerkin_candidate_rhs}
\end{equation}
Subtracting
Eq.~\eqref{eq:galerkin_candidate_rhs}
from
Eq.~\eqref{eq:galerkin_exact_rhs}
gives
\begin{equation}
    B_m(\mathbf y-\mathbf y_z)
    =
    Q_m^\dagger A(x-z).
\end{equation}
Hence
\begin{equation}
    \boxed{
    \mathbf y-\mathbf y_z
    =
    B_m^{-1}Q_m^\dagger A(x-z).}
    \label{eq:galerkin_coordinate_difference}
\end{equation}

Thus, the difference between the Galerkin solution and the accurate
candidate is generated by the subspace-approximation error $x-z$ through
the reduced inverse $B_m^{-1}$.

\paragraph{Step 3: Bound the directional error between $x_m$ and $z$.}

For normalized-state preparation, only the component transverse to the
Galerkin solution changes its direction.  Define
\begin{equation}
    P_{\mathbf y}^{\perp}
    :=
    I-
    \frac{\mathbf y\mathbf y^\dagger}{\|\mathbf y\|_2^2},
\end{equation}
so that
\begin{equation}
    P_{\mathbf y}^{\perp}\mathbf y=0.
\end{equation}
Applying $P_{\mathbf y}^{\perp}$ to
Eq.~\eqref{eq:galerkin_coordinate_difference} yields
\begin{align}
    P_{\mathbf y}^{\perp}\mathbf y_z
    &=
    -P_{\mathbf y}^{\perp}
    B_m^{-1}Q_m^\dagger A(x-z).
\end{align}
Therefore
\begin{align}
    \|P_{\mathbf y}^{\perp}\mathbf y_z\|_2
    &\le
    \|P_{\mathbf y}^{\perp}B_m^{-1}\|_2
    \,
    \|Q_m^\dagger A(x-z)\|_2
    \nonumber\\
    &\le
    K_B\|x-z\|_2.
    \label{eq:galerkin_transverse_bound}
\end{align}
Here we used
\begin{equation}
    K_B
    =
    \|P_{\mathbf y}^{\perp}B_m^{-1}\|_2,
\end{equation}
together with
$\|Q_m^\dagger\|_2=1$ and the normalization
$\|A\|_2\le1$ adopted in the main text.

By assumption,
\begin{equation}
    \|x-z\|_2
    \le
    \epsilon_K\|x\|_2.
    \label{eq:galerkin_candidate_error_recall}
\end{equation}
Since $\epsilon_K<1$, the reverse triangle inequality gives
\begin{align}
    \|z\|_2
    &\ge
    \|x\|_2-\|x-z\|_2
    \nonumber\\
    &\ge
    (1-\epsilon_K)\|x\|_2.
    \label{eq:galerkin_candidate_norm_lower}
\end{align}
Using the isometry of $Q_m$,
$\|\mathbf y_z\|_2=\|z\|_2$, and hence
\begin{align}
    \sin\Theta(x_m,z)
    &=
    \sin\Theta(\mathbf y,\mathbf y_z)
    \nonumber\\
    &=
    \frac{
    \|P_{\mathbf y}^{\perp}\mathbf y_z\|_2
    }{
    \|\mathbf y_z\|_2
    }
    \nonumber\\
    &\le
    \frac{
    K_B\|x-z\|_2
    }{
    (1-\epsilon_K)\|x\|_2
    }
    \nonumber\\
    &\le
    \boxed{
    \frac{K_B\epsilon_K}{1-\epsilon_K}.}
    \label{eq:galerkin_candidate_angle_proof}
\end{align}
This proves
Eq.~\eqref{eq:galerkin_candidate_angle}.

\paragraph{Step 4: Transfer the bound from the candidate to the exact solution.}

It remains to compare the candidate $z$ with the exact solution $x$.
Let $P_z^\perp$ denote the orthogonal projector onto the complement of
$\operatorname{span}\{z\}$.  Since $P_z^\perp z=0$,
\begin{equation}
    P_z^\perp x
    =
    P_z^\perp(x-z),
\end{equation}
and therefore
\begin{align}
    \sin\Theta(z,x)
    &=
    \frac{\|P_z^\perp x\|_2}{\|x\|_2}
    \nonumber\\
    &\le
    \frac{\|x-z\|_2}{\|x\|_2}
    \nonumber\\
    &\le
    \epsilon_K.
    \label{eq:galerkin_candidate_exact_angle}
\end{align}
Hence
\begin{equation}
    \Theta(z,x)
    \le
    \arcsin(\epsilon_K).
\end{equation}

The phase-invariant angle obeys the triangle inequality, so
\begin{equation}
    \Theta(x_m,x)
    \le
    \Theta(x_m,z)
    +
    \Theta(z,x).
\end{equation}
Using
Eq.~\eqref{eq:galerkin_candidate_angle_proof}
and
Eq.~\eqref{eq:galerkin_candidate_exact_angle},
we obtain
\begin{equation}
    \boxed{
    \Theta(x_m,x)
    \le
    \arcsin\!\left(
    \frac{K_B\epsilon_K}{1-\epsilon_K}
    \right)
    +
    \arcsin(\epsilon_K),
    }
    \label{eq:galerkin_transfer_bound_proof}
\end{equation}
provided
$K_B\epsilon_K/(1-\epsilon_K)\le1$.

Finally, if $K_B=O(1)$ and $\epsilon_K\to0$, then
$\arcsin t=t+O(t^3)$ gives
\begin{equation}
    \Theta(x_m,x)
    =
    O(\epsilon_K).
\end{equation}
Thus a sufficiently accurate representation of the exact solution inside
the QKS subspace is inherited by the actual Galerkin output whenever the
directional sensitivity $K_B$ remains bounded.
\end{proof}

\subsection{Projected-measurement precision}
\label{app:measurement_precision_proof}

\begin{proof}[Proof of Corollary~\ref{cor:no_inverse_precision}]
Theorem~\ref{thm:SHg_Bc} converts errors in the actually estimated
quantities into perturbations of the orthonormal reduced system:
\begin{equation}
    \|\widehat B_m-B_m\|_2
    \le
    \frac{\|\Delta H_m\|_2}{\sigma_S}
    =
    \eta_B^{\rm eff},
\end{equation}
and
\begin{equation}
    \|\widehat{\mathbf c}_m-\mathbf c_m\|_2
    \le
    \frac{\|\Delta\mathbf g_m\|_2}{\sqrt{\sigma_S}}
    =
    \eta_c^{\rm eff}.
\end{equation}
If these effective errors satisfy the sufficient finite-stability
conditions of Theorem~\ref{thm:finite_stability}, then
Eq.~\eqref{eq:finite_state_bound} follows directly.

When $\sigma_S,\Delta_B,\beta_B=\Omega(1)$, the coefficients relating
$(\Delta H_m,\Delta\mathbf g_m)$ to
$(\eta_B^{\rm eff},\eta_c^{\rm eff})$
are independent of $\mu_1$.  Hence the additive precisions required for
$H_m$ and $\mathbf g_m$ can be chosen independently of the retained
vanishing eigenvalue, so projected-system estimation acquires no
inverse-$\mu_1$ dependence through its precision requirement.
\end{proof}

\section{End-to-End Composition and Access-Model Details}
\label{app:e2e_and_access}

\subsection{Proof of the combined end-to-end theorem}
\label{app:combined_proof}

\begin{proof}[Proof of Theorem~\ref{thm:combined}]
Theorem~\ref{thm:structural} supplies a vector
$z_m\in\mathcal K_m$ with relative Euclidean error at most $\epsilon_K$
and, under the fixed structural parameters of that theorem,
\[
m=O(1),\qquad T=O(1),\qquad p_{\rm succ}^{-1}=O(1)
\]
with respect to the retained vanishing eigenvalue.
Theorem~\ref{thm:solution_aligned} gives
\[
K_B=O(1),\qquad K_c=O(\mu_1),
\]
although $\|B_m^{-1}\|_2=\Theta(\mu_1^{-1})$.
Theorem~\ref{thm:galerkin_transfer} then transfers the subspace
approximation to the Galerkin output, giving
$\Theta(x_m,x)=O(\epsilon_K)$.

For finite projected-system errors,
Theorem~\ref{thm:SHg_Bc} maps errors in
$(H_m,\mathbf g_m)$ to perturbations of
$(B_m,\mathbf c_m)$ with constants independent of $\mu_1$, and
Theorem~\ref{thm:finite_stability} shows that those perturbations do not
create an inverse-$\mu_1$ normalized-state sensitivity.
Corollary~\ref{cor:no_inverse_precision} therefore permits the projected
quantities to be estimated to fixed target-dependent additive precision
without requiring a precision proportional to $\mu_1$.

Substituting these bounds into the general QKS complexity decomposition
of Theorem~\ref{thm:general_complexity} controls every intrinsic
condition-sensitive channel.  Any remaining inverse dependence on
$\mu_1$ could therefore arise only from the chosen access model or
estimator itself, which is excluded by the final hypothesis of the theorem.
\end{proof}

\subsection{Proof of the end-to-end complexity corollary}
\label{app:e2e_corollary_proof}

\begin{proof}[Proof of Corollary~\ref{cor:e2e}]
Under Theorem~\ref{thm:combined},
$m$, $T$, $p_{\rm succ}^{-1}$, and $K_B$ are bounded with respect to
$\mu_1$, while $K_c$ decreases as $O(\mu_1)$.
The projected-system precision requirement is independent of inverse
powers of $\mu_1$ by Corollary~\ref{cor:no_inverse_precision}.
Hence the general complexity expression reduces, with respect to the
retained-small-eigenvalue scaling, to the remaining projected-system
implementation cost $C_{\rm proj}^{\rm dir}(\epsilon)$ plus bounded
structural, classical-solve, and reconstruction terms.  This gives
Eq.~\eqref{eq:e2e_reduced_complexity}.
\end{proof}

\subsection{Detailed derivations for the access-model comparison}
\label{app:access_model_derivations}

This appendix records the counting used in
Section~\ref{subsec:access_models}.  The purpose is to compare
implementation-dependent costs after the condition-sensitive precision
requirements have already been controlled by
Corollary~\ref{cor:no_inverse_precision}.

\paragraph{Entrywise precision from norm precision.}
If every entry of an $m\times m$ matrix perturbation has magnitude at most
$\xi$, then
\begin{equation}
 \|\Delta M\|_2
 \le
 \|\Delta M\|_\infty
 \le
 m\xi.
\end{equation}
Thus $\|\Delta H_m\|_2=O(\epsilon)$ is guaranteed by
$\xi_H=\Theta(\epsilon/m)$.  Likewise
\begin{equation}
 \|\Delta\mathbf g_m\|_2
 \le
 \sqrt m\,\max_j|\Delta g_j|,
\end{equation}
so $\xi_g=\Theta(\epsilon/\sqrt m)$ suffices.  Because
$H_m$ depends only on the lag $k-j$, only $O(m)$ distinct projected-operator
entries must be estimated.

\paragraph{Block encoding with direct sampling.}
A block encoding of normalization $\alpha$ exposes matrix elements of
$A/\alpha$.  Estimating one $H_m$ lag to absolute accuracy $\xi_H$ by a
Hadamard test therefore requires
\begin{equation}
 O\!\left(\frac{\alpha^2}{\xi_H^2}\right)
\end{equation}
samples.  Across $O(m)$ lags this gives
\begin{equation}
 O\!\left(\frac{\alpha^2m^3}{\epsilon^2}\right).
\end{equation}
The corresponding $\mathbf g_m$ count is
$O(m^2/\epsilon^2)$.  Multiplication by the cost of preparing the relevant
time-evolved states gives
Eq.~\eqref{eq:block_encoding_sampling_cost}.

\paragraph{Block encoding with coherent estimation.}
Replacing independent sampling by a coherent $O(1/\xi)$ expectation-value
estimator changes the two query counts to
\begin{equation}
 O\!\left(\frac{\alpha m^2}{\epsilon}\right),
 \qquad
 O\!\left(\frac{m^{3/2}}{\epsilon}\right),
\end{equation}
which yields Eq.~\eqref{eq:block_encoding_coherent_cost}.  This option uses
deeper coherent circuits and stronger access assumptions than direct
sampling; it is included to distinguish access-model dependence from
estimator dependence.

\paragraph{Sparse access with short-time finite differences.}
For Eq.~\eqref{eq:xu_access_model_estimator}, an element error
$\xi_H$ has the form
\begin{equation}
 \xi_H
 =
 O(h)
 +
 O(\delta_H/h)
 +
 O(\eta_{\rm sim}/h).
\end{equation}
Balancing the first two terms gives
$h=\Theta(\xi_H)$ and
$\delta_H=\Theta(\xi_H^2)$; choosing
$\eta_{\rm sim}=O(\xi_H^2)$ adds only polylogarithmic simulation-precision
overhead.  Coherent estimation of the $O(m)$ distinct dynamical
correlations therefore requires
\begin{equation}
 \widetilde O\!\left(
 \frac{\sqrt m}{\xi_H^2}
 \right)
 =
 \widetilde O\!\left(
 \frac{m^{5/2}}{\epsilon^2}
 \right)
\end{equation}
oracle calls.  Direct estimation of $\mathbf g_m$ requires
$\widetilde O(m/\epsilon)$ calls.  Multiplying by the
normalized sparse-simulation cost
$\widetilde O(sT\,\operatorname{polylog}N)$ gives
Eq.~\eqref{eq:sparse_fd_access_cost}.

\paragraph{Pauli/LCU direct projection.}
For Eq.~\eqref{eq:pauli_access_decomposition}, estimating every term-level
correlation to accuracy $\delta_P$ produces an $H_m$-element error at most
$\Lambda_A\delta_P$.  Hence
\begin{equation}
 \delta_P
 =
 \Theta\!\left(
 \frac{\epsilon}{m\Lambda_A}
 \right).
\end{equation}
There are $O(mL)$ correlations.  The coherent multiple-expectation
estimator of Ref.~\cite{Huggins2022} therefore uses
\begin{equation}
 \widetilde O\!\left(
 \frac{\sqrt{mL}}{\delta_P}
 \right)
 =
 \widetilde O\!\left(
 \frac{\Lambda_A m^{3/2}\sqrt L}{\epsilon}
 \right)
\end{equation}
coherent oracle calls.  The $\mathbf g_m$ cost is
$\widetilde O(m/\epsilon)$.  If real-time evolution is implemented from the
same LCU representation by qubitization, the per-call evolution cost is
$\widetilde O(\Lambda_A T+1)$, yielding
Eq.~\eqref{eq:pauli_access_cost}.

\paragraph{Interpretation.}
These derivations should not be read as an optimality claim for any access
model.  They provide representative end-to-end realizations that expose
which parameters are implementation dependent.  In particular, the
short-time estimator pays an additional finite-difference precision cost,
whereas direct block-encoded or Pauli matrix-element estimation avoids that
penalty.  Across all cases, the new directional stability theorem removes
the need for the projected-system precision itself to scale as
$O(\mu_1)$.

\subsection{Proof of access-model robustness}
\label{app:access_robustness_proof}

\begin{proof}[Proof of Corollary~\ref{cor:access_model_robustness}]
Theorem~\ref{thm:combined} gives $m,T=O(1)$ with respect to $\mu_1$,
and Corollary~\ref{cor:no_inverse_precision} makes the required
projected-system accuracy independent of inverse powers of $\mu_1$.
Inspection of
Eqs.~\eqref{eq:block_encoding_sampling_cost},
\eqref{eq:block_encoding_coherent_cost},
\eqref{eq:sparse_fd_access_cost}, and
\eqref{eq:pauli_access_cost}
then shows that the only remaining factors are the target accuracy and the
parameters of the selected access model.  If those parameters do not
themselves scale inversely with $\mu_1$, none of the displayed
implementation costs acquires such a dependence.
\end{proof}

\section{Proofs for the Relation to Existing Beyond-Conditioning Methods}
\label{app:separation_proofs}

\subsection{Truncatable-regime theorem}
\label{app:truncatable_proof}

\begin{proof}[Proof of Theorem~\ref{thm:truncatable_branch}]

The purpose of this proof is to show that, when all solution-relevant
active eigenvalues remain bounded away from zero, the QKS subspace has only
small leakage into the inactive spectral sector.  This prevents arbitrarily
small inactive eigenvalues from making the reduced QKS problem ill
conditioned.

We use the standing normalization
\begin{equation}
    \|A\|_2\le1.
\end{equation}

\paragraph{Step 1: Convert solution-tail control into inactive input-weight control.}

Recall that
\begin{equation}
    x=A^{-1}|b\rangle
    =
    \sum_{\ell=1}^{N}
    \frac{b_\ell}{\lambda_\ell}|u_\ell\rangle.
\end{equation}
Assumption (A1) gives
\begin{equation}
    \|P_{\mathcal I}x\|_2
    \le
    \epsilon_{\rm sub}\|x\|_2.
    \label{eq:app_trunc_solution_tail}
\end{equation}

Since $\lambda_\ell\le1$ under the normalization $\|A\|_2\le1$,
\begin{align}
    \|P_{\mathcal I}b\|_2^2
    &=
    \sum_{\ell\in\mathcal I}|b_\ell|^2
    \nonumber\\
    &=
    \sum_{\ell\in\mathcal I}
    \lambda_\ell^2
    \frac{|b_\ell|^2}{\lambda_\ell^2}
    \nonumber\\
    &\le
    \sum_{\ell\in\mathcal I}
    \frac{|b_\ell|^2}{\lambda_\ell^2}
    =
    \|P_{\mathcal I}x\|_2^2.
\end{align}
Therefore
\begin{equation}
    \|P_{\mathcal I}b\|_2
    \le
    \epsilon_{\rm sub}\|x\|_2.
    \label{eq:app_inactive_b_first}
\end{equation}

To eliminate $\|x\|_2$ from this bound, use the assumption
\begin{equation}
    \lambda_a\ge\Delta_{\mathcal A},
    \qquad
    a\in\mathcal A.
\end{equation}
Then
\begin{align}
    \|P_{\mathcal A}b\|_2^2
    &=
    \sum_{a\in\mathcal A}
    \lambda_a^2
    \frac{|b_a|^2}{\lambda_a^2}
    \nonumber\\
    &\ge
    \Delta_{\mathcal A}^2
    \sum_{a\in\mathcal A}
    \frac{|b_a|^2}{\lambda_a^2}.
\end{align}
By the definition of the active set,
\begin{equation}
    \sum_{a\in\mathcal A}
    \frac{|b_a|^2}{\lambda_a^2}
    \ge
    (1-\epsilon_{\rm sub}^2)\|x\|_2^2.
\end{equation}
Hence
\begin{equation}
    \|P_{\mathcal A}b\|_2
    \ge
    \Delta_{\mathcal A}
    \sqrt{1-\epsilon_{\rm sub}^2}
    \|x\|_2.
\end{equation}
Since $|b\rangle$ is normalized,
$\|P_{\mathcal A}b\|_2\le1$, and therefore
\begin{equation}
    \|x\|_2
    \le
    \frac{1}{
    \Delta_{\mathcal A}
    \sqrt{1-\epsilon_{\rm sub}^2}}.
    \label{eq:app_x_upper_trunc}
\end{equation}

Combining
Eqs.~\eqref{eq:app_inactive_b_first}
and~\eqref{eq:app_x_upper_trunc} gives
\begin{equation}
    \boxed{
    \|P_{\mathcal I}b\|_2
    \le
    \frac{
    \epsilon_{\rm sub}
    }{
    \Delta_{\mathcal A}
    \sqrt{1-\epsilon_{\rm sub}^2}
    }.}
    \label{eq:app_inactive_b_bound}
\end{equation}

Thus small normalized-solution weight outside the active sector implies
small inactive amplitude in the input state whenever the active spectrum
remains separated from zero.

\paragraph{Step 2: Bound the inactive part of the time-evolution Krylov basis.}

The $j$th QKS basis vector is
\begin{equation}
    |\phi_j\rangle
    =
    e^{-ij\tau A}|b\rangle.
\end{equation}
Because $P_{\mathcal I}$ is a spectral projector of $A$,
it commutes with the time evolution.  Therefore
\begin{equation}
    P_{\mathcal I}|\phi_j\rangle
    =
    e^{-ij\tau A}P_{\mathcal I}|b\rangle,
\end{equation}
and hence
\begin{equation}
    \|P_{\mathcal I}|\phi_j\rangle\|_2
    =
    \|P_{\mathcal I}b\|_2.
\end{equation}

Writing
\begin{equation}
    \Phi_m
    =
    [|\phi_0\rangle,\ldots,|\phi_{m-1}\rangle],
\end{equation}
we obtain
\begin{align}
    \|P_{\mathcal I}\Phi_m\|_2
    &\le
    \|P_{\mathcal I}\Phi_m\|_F
    \nonumber\\
    &=
    \sqrt m\,
    \|P_{\mathcal I}b\|_2.
    \label{eq:app_inactive_Phi}
\end{align}

\paragraph{Step 3: Transfer the bound through orthonormalization.}

The orthonormal QKS basis is
\begin{equation}
    Q_m
    =
    \Phi_mS_m^{-1/2}.
\end{equation}
Assumption (A2), together with the structural theorem, gives
\begin{equation}
    \lambda_{\min}(S_m)
    \ge
    s_0^2,
\end{equation}
and consequently
\begin{equation}
    \|S_m^{-1/2}\|_2
    \le
    \frac1{s_0}.
\end{equation}

Therefore
\begin{align}
    \eta_{\rm leak}
    :=
    \|P_{\mathcal I}Q_m\|_2
    &\le
    \|P_{\mathcal I}\Phi_m\|_2
    \|S_m^{-1/2}\|_2
    \nonumber\\
    &\le
    \frac{\sqrt m}{s_0}
    \|P_{\mathcal I}b\|_2.
\end{align}
Using Eq.~\eqref{eq:app_inactive_b_bound},
\begin{equation}
    \boxed{
    \eta_{\rm leak}
    \le
    \frac{
    \sqrt m\,\epsilon_{\rm sub}
    }{
    s_0\Delta_{\mathcal A}
    \sqrt{1-\epsilon_{\rm sub}^2}
    }.}
    \label{eq:app_leakage_final}
\end{equation}

This proves the first claim.

\paragraph{Step 4: Convert small subspace leakage into a reduced spectral gap.}

Let $\mathbf y\in\mathbb C^m$ satisfy
$\|\mathbf y\|_2=1$, and define
\begin{equation}
    |q\rangle
    :=
    Q_m\mathbf y.
\end{equation}
Since $Q_m^\dagger Q_m=I$,
\begin{equation}
    \|q\|_2=1.
\end{equation}
Moreover,
\begin{equation}
    \|P_{\mathcal I}q\|_2
    \le
    \eta_{\rm leak}
    \le
    \eta_0,
\end{equation}
so
\begin{equation}
    \|P_{\mathcal A}q\|_2^2
    =
    1-\|P_{\mathcal I}q\|_2^2
    \ge
    1-\eta_0^2.
\end{equation}

Because every active eigenvalue satisfies
$\lambda_a\ge\Delta_{\mathcal A}$ and $A$ is positive definite,
\begin{align}
    \langle q|A|q\rangle
    &\ge
    \Delta_{\mathcal A}
    \|P_{\mathcal A}q\|_2^2
    \nonumber\\
    &\ge
    \Delta_{\mathcal A}
    (1-\eta_0^2).
\end{align}
Using
\begin{equation}
    B_m=Q_m^\dagger A Q_m,
\end{equation}
we therefore have
\begin{equation}
    \mathbf y^\dagger B_m\mathbf y
    \ge
    \Delta_{\mathcal A}
    (1-\eta_0^2)
\end{equation}
for every unit vector $\mathbf y$.

The min--max principle then gives
\begin{equation}
    \boxed{
    \lambda_{\min}(B_m)
    \ge
    \Delta_{\mathcal A}
    (1-\eta_0^2).}
\end{equation}

Hence globally small eigenvalues outside the active spectral sector can
drive the global condition number to infinity without making the QKS
reduced system ill conditioned.
\end{proof}

\subsection{Proof of the overlap corollary}
\label{app:overlap_corollary_proof}

\begin{proof}[Proof of Corollary~\ref{cor:overlap_truncation}]
Theorem~\ref{thm:truncatable_branch} gives
$\lambda_{\min}(B_m)=\Omega(1)$, and therefore
\[
K_B\le \|B_m^{-1}\|_2=O(1).
\]
Theorem~\ref{thm:structural} simultaneously gives
$m,T,p_{\rm succ}^{-1}=O(1)$ under the structural assumptions.
Thus the global condition number may diverge through inactive eigenvalues
while every listed QKS quantity remains bounded with respect to those
eigenvalues.
\end{proof}

\subsection{Explicit separation}
\label{app:separation_su_proof}

\begin{proof}[Proof of Theorem~\ref{thm:separation_su}]

The purpose of this construction is to exhibit a family for which full
inverse amplification diverges, while the QKS structural and normalized
directional quantities remain bounded.

\paragraph{Step 1: Global conditioning and the normalized solution.}

Consider
\begin{equation}
    A_\varepsilon
    =
    \begin{pmatrix}
        \varepsilon & 0\\
        0 & 1
    \end{pmatrix},
    \qquad
    |b\rangle
    =
    \frac1{\sqrt2}
    \begin{pmatrix}
        1\\
        1
    \end{pmatrix}.
\end{equation}
Since the eigenvalues are $\varepsilon$ and $1$,
\begin{equation}
    \kappa(A_\varepsilon)
    =
    \varepsilon^{-1}.
\end{equation}

Moreover,
\begin{equation}
    A_\varepsilon^{-1}|b\rangle
    =
    \frac1{\sqrt2}
    \begin{pmatrix}
        \varepsilon^{-1}\\
        1
    \end{pmatrix}.
\end{equation}
Its norm is
\begin{equation}
    \|A_\varepsilon^{-1}|b\rangle\|_2
    =
    \frac1{\sqrt2\varepsilon}
    \sqrt{1+\varepsilon^2}.
\end{equation}
Therefore the normalized solution is
\begin{equation}
    |x_\varepsilon\rangle
    =
    \frac{
    (1,\varepsilon)^T
    }{
    \sqrt{1+\varepsilon^2}
    },
\end{equation}
which approaches the soft eigendirection as
$\varepsilon\to0$.

\paragraph{Step 2: The full inverse-amplification quantity diverges.}

Since $A_\varepsilon$ is Hermitian,
\begin{equation}
    A_\varepsilon^{-1\dagger}
    =
    A_\varepsilon^{-1}.
\end{equation}
Applying the inverse to the normalized solution gives
\begin{equation}
    A_\varepsilon^{-1}|x_\varepsilon\rangle
    =
    \frac1{\sqrt{1+\varepsilon^2}}
    \begin{pmatrix}
        \varepsilon^{-1}\\
        \varepsilon
    \end{pmatrix}.
\end{equation}
Hence
\begin{equation}
    \chi_{\rm filt}(A_\varepsilon,b)
    =
    \frac{
    \sqrt{\varepsilon^{-2}+\varepsilon^2}
    }{
    \sqrt{1+\varepsilon^2}
    }.
\end{equation}
As $\varepsilon\to0$,
\begin{equation}
    \boxed{
    \chi_{\rm filt}(A_\varepsilon,b)
    =
    \Theta(\varepsilon^{-1}).}
\end{equation}

Thus the full inverse-amplification criterion inherits the divergent small
eigenvalue.

\paragraph{Step 3: The two-state QKS basis remains uniformly stable.}

Choose a fixed
\begin{equation}
    \tau\notin2\pi\mathbb Z.
\end{equation}
For $m=2$, the QKS basis is
\begin{equation}
    \Phi_2
    =
    \left[
    |b\rangle,\;
    e^{-i\tau A_\varepsilon}|b\rangle
    \right],
\end{equation}
with
\begin{equation}
    e^{-i\tau A_\varepsilon}|b\rangle
    =
    \frac1{\sqrt2}
    \begin{pmatrix}
        e^{-i\tau\varepsilon}\\
        e^{-i\tau}
    \end{pmatrix}.
\end{equation}

The two columns are linearly independent whenever
\begin{equation}
    e^{-i\tau\varepsilon}
    \neq
    e^{-i\tau}.
\end{equation}
As $\varepsilon\to0$, their determinant approaches a nonzero multiple of
\begin{equation}
    e^{-i\tau}-1,
\end{equation}
which is nonzero by the choice of $\tau$.

More explicitly, the Gram matrix has the form
\begin{equation}
    S_2
    =
    \begin{pmatrix}
        1 & g_\varepsilon\\
        g_\varepsilon^* & 1
    \end{pmatrix},
\end{equation}
where
\begin{equation}
    g_\varepsilon
    =
    \frac{
    e^{-i\tau\varepsilon}+e^{-i\tau}
    }{2}.
\end{equation}
Therefore
\begin{equation}
    \lambda_{\min}(S_2)
    =
    1-|g_\varepsilon|.
\end{equation}
Using
\begin{equation}
    |g_\varepsilon|
    =
    \left|
    \cos\!\left(
    \frac{\tau(1-\varepsilon)}{2}
    \right)
    \right|,
\end{equation}
we obtain
\begin{equation}
    \lambda_{\min}(S_2)
    \longrightarrow
    1-
    \left|\cos\!\left(\frac{\tau}{2}\right)\right|
    >0.
\end{equation}
Hence, for sufficiently small $\varepsilon$,
\begin{equation}
    \boxed{
    \lambda_{\min}(S_2)=\Omega(1).}
\end{equation}

Thus the QKS basis does not become singular as the physical eigenvalue
$\varepsilon$ vanishes.

\paragraph{Step 4: The QKS subspace captures the full relevant space with bounded structural resources.}

Since the two basis vectors are linearly independent in a
two-dimensional Hilbert space,
\begin{equation}
    \mathcal K_2
    =
    \mathbb C^2.
\end{equation}
Therefore
\begin{equation}
    m=2,
    \qquad
    T=\tau=O(1).
\end{equation}

The Gram lower bound also gives
\begin{equation}
    p_{\rm succ}
    \ge
    \frac{\lambda_{\min}(S_2)}{m},
\end{equation}
and hence
\begin{equation}
    \boxed{
    p_{\rm succ}^{-1}=O(1).}
\end{equation}

\paragraph{Step 5: The reduced system retains the small eigenvalue.}

Let
\begin{equation}
    Q_2
    =
    \Phi_2S_2^{-1/2}.
\end{equation}
Because $Q_2$ is a square orthonormal matrix, it is unitary.  Therefore
\begin{equation}
    B_2
    =
    Q_2^\dagger A_\varepsilon Q_2
\end{equation}
is unitarily similar to $A_\varepsilon$ and has eigenvalues
\begin{equation}
    \mu_1=\varepsilon,
    \qquad
    \mu_2=1.
\end{equation}

Let $|u_1\rangle,|u_2\rangle$ be the eigenvectors of
$A_\varepsilon$, and define
\begin{equation}
    \mathbf v_i
    :=
    Q_2^\dagger|u_i\rangle.
\end{equation}
These are eigenvectors of $B_2$.

The orthonormal reduced right-hand side is
\begin{equation}
    \mathbf c_2
    =
    Q_2^\dagger|b\rangle.
\end{equation}
Therefore
\begin{align}
    |\langle\mathbf v_i,\mathbf c_2\rangle|
    &=
    |\langle u_i|Q_2Q_2^\dagger|b\rangle|
    \nonumber\\
    &=
    |\langle u_i|b\rangle|
    =
    \frac1{\sqrt2}.
\end{align}
Hence the small eigenvalue is fully retained and remains
solution-relevant:
\begin{equation}
    |c_1|
    =
    \frac1{\sqrt2}
    =
    \Omega(1).
\end{equation}

\paragraph{Step 6: Directional sensitivity remains bounded.}

The hypotheses of
Theorem~\ref{thm:solution_aligned}
are satisfied with
\begin{equation}
    \mu_1=\varepsilon,
    \qquad
    \mu_2=1,
    \qquad
    |c_1|=\Omega(1).
\end{equation}
Therefore
\begin{equation}
    \boxed{
    K_B=O(1),
    \qquad
    K_c=O(\varepsilon).}
\end{equation}

We have therefore constructed a family for which
\begin{equation}
    \chi_{\rm filt}
    =
    \Theta(\varepsilon^{-1}),
\end{equation}
while
\begin{equation}
    m,\;
    T,\;
    p_{\rm succ}^{-1},\;
    K_B
    =
    O(1),
    \qquad
    K_c=O(\varepsilon).
\end{equation}

Thus full inverse amplification can diverge even when the normalized
directional QKS mechanism remains uniformly stable.
\end{proof}

\subsection{High-dimensional embedding}
\label{app:highdim_proof}

\begin{proof}[Proof of Proposition~\ref{prop:high_dimensional_separation}]

The purpose of the construction is to show that the separation in
Theorem~\ref{thm:separation_su} persists when the QKS subspace is much
smaller than the ambient Hilbert space.

Let
\begin{equation}
    A_{\varepsilon,N}
    =
    A_\varepsilon\oplus D_{N-2},
    \qquad
    |b_N\rangle
    =
    |b\rangle\oplus0,
\end{equation}
where
\begin{equation}
    \operatorname{spec}(D_{N-2})
    \subseteq
    [\Delta_{\rm sp},1]
\end{equation}
for a fixed $\Delta_{\rm sp}>0$.

\paragraph{Step 1: The global condition number still diverges.}

For sufficiently small
\begin{equation}
    \varepsilon<\Delta_{\rm sp},
\end{equation}
the smallest eigenvalue of $A_{\varepsilon,N}$ is $\varepsilon$, while
its largest eigenvalue is $1$.  Therefore
\begin{equation}
    \boxed{
    \kappa(A_{\varepsilon,N})
    =
    \varepsilon^{-1}
    =
    \Theta(\varepsilon^{-1}).}
\end{equation}

\paragraph{Step 2: The target solution remains entirely inside the first block.}

Since the right-hand side has no support on the additional block,
\begin{equation}
    A_{\varepsilon,N}^{-1}|b_N\rangle
    =
    \left(
    A_\varepsilon^{-1}|b\rangle
    \right)
    \oplus0.
\end{equation}
After normalization,
\begin{equation}
    |x_{\varepsilon,N}\rangle
    =
    |x_\varepsilon\rangle\oplus0.
\end{equation}
Applying the inverse again gives
\begin{equation}
    A_{\varepsilon,N}^{-1}
    |x_{\varepsilon,N}\rangle
    =
    \left(
    A_\varepsilon^{-1}|x_\varepsilon\rangle
    \right)
    \oplus0.
\end{equation}
Hence
\begin{equation}
    \boxed{
    \chi_{\rm filt}(A_{\varepsilon,N},b_N)
    =
    \chi_{\rm filt}(A_\varepsilon,b)
    =
    \Theta(\varepsilon^{-1}).}
\end{equation}

Thus the full inverse-amplification divergence is unchanged by the
embedding.

\paragraph{Step 3: The QKS dynamics remain confined to a two-dimensional invariant subspace.}

For every evolution time $t$,
\begin{align}
    e^{-itA_{\varepsilon,N}}|b_N\rangle
    &=
    \left(
    e^{-itA_\varepsilon}|b\rangle
    \right)
    \oplus0.
\end{align}
Therefore every time-evolution Krylov state lies inside
\begin{equation}
    \mathcal H_{\rm rel}
    :=
    \mathbb C^2\oplus0,
\end{equation}
which is an invariant subspace of $A_{\varepsilon,N}$.

Consequently, the QKS Gram matrix, projected operator, and projected
right-hand side are identical to those of the two-dimensional construction
in Theorem~\ref{thm:separation_su}.  In particular, for the same fixed
choice of $\tau$,
\begin{equation}
    m=2,
    \qquad
    T=O(1),
    \qquad
    p_{\rm succ}^{-1}=O(1).
\end{equation}

The equivalent reduced system also has
\begin{equation}
    \mu_1=\varepsilon,
    \qquad
    \mu_2=1,
\end{equation}
with nonvanishing right-hand-side overlap on the soft eigendirection.
Therefore
\begin{equation}
    K_B=O(1),
    \qquad
    K_c=O(\varepsilon).
\end{equation}

Hence
\begin{equation}
    \boxed{
    m=2\ll N,
    \qquad
    T=O(1),
    \qquad
    p_{\rm succ}^{-1}=O(1),
    \qquad
    K_B=O(1),
    \qquad
    K_c=O(\varepsilon).}
\end{equation}

The separation between full inverse amplification and normalized
directional sensitivity therefore persists in arbitrarily large ambient
dimension and is not an artifact of the QKS subspace coinciding with the
full Hilbert space.
\end{proof}

\bibliographystyle{unsrt}
\bibliography{fullreferences}
\end{document}